\documentclass[conference]{IEEEtran}
\usepackage{cite}
\usepackage{graphicx}
\usepackage{psfrag}
\usepackage{url}
\usepackage{amsmath}
\usepackage{array}
\usepackage{amssymb}
\usepackage{amsfonts}
\usepackage{epstopdf}

\newtheorem{lemma}{Lemma}

\newtheorem{definition}{Definition}
\newtheorem{theorem}{Theorem}

\usepackage{subfigure}

\title{Interference Alignment with Diversity for the $2 \times 2$ $X$ Network with four antennas }

\begin{document}

\author{
\authorblockN{Abhinav Ganesan and B. Sundar Rajan\\}
\IEEEauthorblockA{\small{Email: \{abhig\_88, bsrajan\}@ece.iisc.ernet.in}}
}

\maketitle
\thispagestyle{empty}	

\begin{abstract}
A transmission scheme based on the Alamouti code, which we call the Li-Jafarkhani-Jafar (LJJ) scheme, was recently proposed for the $2 \times 2$ $X$ Network (i.e., two-transmitter (Tx) two-receiver (Rx) $X$ Network) with two antennas at each node. This scheme was claimed to achieve a sum degrees of freedom (DoF) of $\frac{8}{3}$ and also a diversity gain of two when fixed finite constellations are employed at each Tx. Furthermore, each Tx required the knowledge of only its own channel unlike the Jafar-Shamai scheme which required global CSIT to achieve the maximum possible sum DoF of $\frac{8}{3}$. In this paper, we extend the LJJ scheme to the $2 \times 2$ $X$ Network with four antennas at each node. The proposed scheme also assumes only local channel knowledge at each Tx. We prove that the proposed scheme achieves the maximum possible sum DoF of $\frac{16}{3}$. In addition, we also prove that, using any fixed finite constellation with appropriate rotation at each Tx, the proposed scheme achieves a diversity gain of at least four.

\end{abstract}	

\section{Introduction}  \label{sec1}
The problem of capacity region of Gaussian interference networks has been open for decades except for a few special cases \cite{HS,RKF}. In the course of pursuit of capacity region of general Gaussian interference networks, researchers have been led into approximating their capacity regions (see for example, \cite{ETW}) and their sum-capacities. A popular way of approximating the sum-capacity of a Gaussian interference network is using the concept of degrees of freedom (DoF). The sum DoF of a Gaussian interference network is said to be $d$ if the sum-capacity can be written as $d~log_2 SNR + o(log_2 SNR)$ \cite{JaS}. A $K \times J$ MIMO $X$ network is a Gaussian interference network where each of the $J$ receivers (Rx) require one independent message from each of the $K$ transmitters (Tx). Henceforth, a $K \times J$ MIMO $X$ network with $M$ antennas at each node shall be abbreviated as $(K,J,M)-X$ Network. The sum DoF of $(2,2,M)-X$ Network was studied in \cite{MMK,JaS}. In \cite{MMK}, it was shown that a sum DoF of $\lfloor{\frac{4M}{3}}\rfloor$ is achievable in a $(2,2,M)-X$ Network while the work in \cite{JaS} shows that a sum DoF of $\frac{4M}{3}$ is achievable. Furthermore, $\frac{4M}{3}$ was also proven to be an outerbound on the sum DoF of $(2,2,M)-X$ Network \cite{JaS}. The transmission scheme in \cite{JaS} that achieved this sum DoF was based on the idea of interference alignment (IA). We shall henceforth call this scheme as the Jafar-Shamai scheme.

The concept of IA for $M>1$ involved linear precoding using a $3$-symbol extension of the channel in such a way that the interference subspaces at the receivers overlap while being linearly independent of the desired signal subspace. This assumed constant channel matrices and knowledge of all the channel gains at both the transmitters (i.e., global CSIT). The desired signals were retrieved by simple zero-forcing. 

In a recent work by Li et al. \cite{LJJ} an IA scheme for $(2,2,2)-X$ Network using the Alamouti code and appropriate channel dependent precoding was proposed. In this scheme,  each transmitter needs the knowledge of the channel from itself to both the receivers (i.e., local CSIT) whereas, in the Jafar-Shamai scheme, global CSIT is needed. This scheme, which we call the LJJ scheme, claimed to achieve the sum DoF of $(2,2,2)-X$ Network which is equal to $\frac{8}{3}$. However, \cite{LJJ} assumed the channel gains to be independently distributed as circularly symmetric complex Gaussian. Also, the proof of achievability of the sum DoF of $(2,2,2)-X$ Network is incomplete. We present a complete proof in Section \ref{subsec2} of this paper with the assumption that the real and imaginary parts of the channel gains are distributed independently according to an arbitrary continuous distribution like in the Jafar-Shamai scheme. Further, the LJJ scheme also achieves a diversity gain of two with node-to-node symbol rate of $\frac{2}{3}$ complex symbols per channel use (cspcu) where, the complex symbols are assumed to take values from a fixed finite constellation.
 
In this work, we extend the LJJ scheme to $(2,2,4)-X$ Network using Srinath-Rajan (S-R) space-time block code (STBC) which was proposed for the asymmetric $4 \times 2$ single user MIMO system \cite{SrR1}. The S-R code possesses a repetitive Alamouti structure upto scaling by a constant. This makes it convenient to adapt the LJJ scheme to $(2,2,4)-X$ Network. We prove that the proposed scheme achieves the sum DoF of $(2,2,4)-X$ Network which is equal to $\frac{16}{3}$. This scheme also requires only local CSIT like the LJJ scheme. Furthermore, under a more practical scenario of fixed finite constellation inputs, we prove that the proposed scheme achieves a diversity gain of at least four.

 The contributions of the paper are summarized below.

\begin{itemize}
\item We provide a complete proof of achievability of sum DoF of $\frac{8}{3}$ by the LJJ scheme (see Theorem \ref{thm3} in Section \ref{subsec2}).
\item We extend the LJJ scheme to $(2,2,4)-X$ Network using the S-R STBC. It is proved that this scheme achieves a sum DoF of $\frac{16}{3}$ (see Theorem \ref{thm5} in Section \ref{sec4}). The proposed scheme requires only local CSIT while the Jafar-Shamai scheme requires global CSIT to achieve the same sum DoF.
\item We prove that the proposed scheme also achieves a diversity gain of at least four (see Theorem \ref{thm4} in Section \ref{sec4}) when fixed finite constellations are employed at the transmitters. Simulation results show that the diversity gain of the proposed scheme is strictly greater than four.
\end{itemize}

The paper is organized as follows. Section \ref{sec2} formally introduces the system model. A brief overview of the Jafar-Shamai scheme for $(2,2,4)-X$ Network and the LJJ scheme for $(2,2,2)-X$ Network along with a complete proof of the sum DoF achieved by the LJJ scheme is given in Section \ref{sec3}. Extension of the LJJ scheme for $(2,2,4)-X$ Network based on the S-R STBC is described in Section \ref{sec4}. Simulation results comparing the proposed scheme with the Jafar-Shamai scheme and the time division multiple access (TDMA) scheme are presented in Section \ref{sec5}. We conclude the paper with Section \ref{sec6}.

{\em Notations:}  The set of complex number is denoted by $\mathbb C$. The notation ${\cal CN}(0,\sigma^2)$ denotes the circularly symmetric complex Gaussian distribution with mean zero and variance $\sigma^2$. For a complex number $x$, the notation $\overline{x}$ denotes the conjugate of $x$. The real and imaginary parts of a complex number $a$ are denoted by $a^R$ and $a^I$ respectively. The trace of a matrix $A$ is denoted by $\text{tr}(A)$. For an invertible matrix $A$, the notation $A^{-H}$ denotes the hermitian of the matrix $A^{-1}$. The $i^\text{th}$ row, $j^\text{th}$ column element of a matrix $A$ is denoted by $a_{ij}$. The $i^\text{th}$ row and the $i^\text{th}$ column of a matrix $A$ are denoted by $A(i,:)$ and $A(:,i)$ respectively. The Frobenius norm of a matrix $A$ is denoted by $||A||$. The identity matrix of size $n \times n$ is denoted by $I_n$. The Kronecker product of two matrices $A$ and $B$ is denoted by $A \otimes B$. A diagonal matrix with the diagonal entries given by $a_1, a_2, \cdots,a_n$ is denoted by $\text{diag}(a_1,a_2,\cdots,a_n)$. The notation $vec(A)$ denotes the vectorized version of the matrix $A$.

\section{System Model} \label{sec2}
 \begin{figure}[htbp]
\centering
\includegraphics[totalheight=3.1in,width=2.8in]{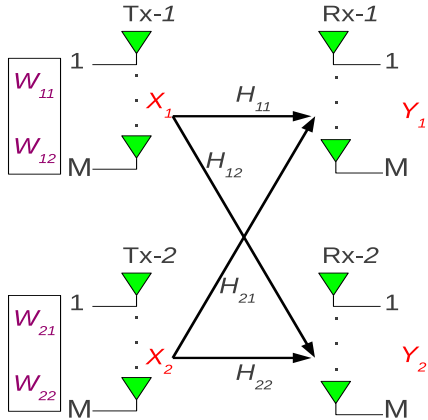}
\vspace{-1cm}
\caption{System Model.}
\label{fig_sys_model}
\end{figure}The $(2,2,M)-X$ Network is shown in Fig. \ref{fig_sys_model}. Each transmitter Tx-$i$ has an independent message $W_{ij}$ for each receiver Rx-$j$, where $i,j=1,2$. The message generated by Tx-$i$ for Rx-$j$ is denoted by $W_{ij}$. The input symbols and the output symbols over $T$ time slots are related as 
\begin{align}
\label{eqn-sys_model}
 Y_j=\sqrt{\frac{P}{M}}\sum_{i=1}^{2} H_{ij}X_i+N_j
\end{align}where, $Y_j \in \mathbb{C}^{M \times T}$ denotes the output matrix at Rx-$j$, $X_i \in \mathbb{C}^{M \times T}$ denotes the input matrix at Tx-$i$ such that $\mathbb{E}\left[\text{tr}\left(XX^H\right)\right]\leq TM$, $H_{ij} \in \mathbb{C}^{M \times M}$ denotes the channel matrix between Tx-$i$ and Rx-$j$, $N_j \in \mathbb{C}^{M \times T}$ denotes the noise matrix whose entries are i.i.d. distributed as ${\cal CN}(0,1)$. As in \cite{JaS}, we assume that the entries of all the channel matrices are independent and take values from arbitrary continuous probability distribution\footnote{We consider a complex random variable to have a continuous probability distribution if its real and imaginary parts are independent and distributed according to some continuous distribution.} so that they are almost surely full rank. Specifically, for the diversity gain evaluations, we assume that the channel matrix entries are distributed as  i.i.d. ${\cal CN}(0,1)$. The channel gains are assumed to be a constant over the transmitted codeword length. All the channel gains are assumed to be known to both the receivers (i.e., global CSIR), and this will not be specifically mentioned henceforth. The average power constraints at both the transmitters are assumed to be equal to $P$. The achievable rates and sum DoF of $(2,2,M)-X$ Network are defined in the conventional sense \cite{JaS}.

\section{Background - Jafar-Shamai Scheme and LJJ Scheme} \label{sec3}
In the first sub-section we shall briefly review the Jafar-Shamai scheme from \cite{JaS} and in the second sub-section we shall review the LJJ scheme from \cite{LJJ}.
\subsection{Review of Jafar-Shamai Scheme for $(2,2,4)-X$ Network} \label{subsec1}
 The Jafar-Shamai scheme for $(2,2,4)-X$ Network aligns the interference symbols by precoding over a $3$-symbol extension of the channel, i.e., $T=3$. Each transmitter transmits $4$ complex symbols to each receiver over $3$ channel uses so that a sum DoF of $\frac{16}{3}$ is achieved. The input-output relation over a $3$-symbol extension of the channel is given by
\begin{align} \label{eqn-JS}
 Y'_j= \sqrt{\frac{3P}{2}}\sum_{i=1}^{2}H'_{ij} \left(\sum_{k=1}^{2}\frac{V_{ik}}{\text{tr}\left(V_{ik}V^H_{ik}\right)} X_{ik}\right) + N'_j
\end{align}where, $Y'_j \in \mathbb C^{12 \times 1}$ denotes the received symbol vector at Rx-$j$ over $3$ channel uses, {\small$H'_{ij}=\begin{bmatrix}
                                                                                                                            H_{ij} & \mathbf{0} & \mathbf{0}\\
															    \mathbf{0} & H_{ij} &\mathbf{0}\\
															    \mathbf{0} & \mathbf{0} & H_{ij}
                                                                                                                           \end{bmatrix}$}
denotes the effective channel matrix between Tx-$i$ and Rx-$j$ over $3$ channel uses, $V_{ik} \in \mathbb{C}^{12 \times 4}$ denotes the precoding matrix, $X_{ik}\in \mathbb C^{4 \times 1}$ denotes the symbol vector generated by Tx-$i$ meant for Rx-$k$, and $N'_j \in \mathbb{C}^{12 \times 1}$ denotes the Gaussian noise vector whose entries are distributed as i.i.d. ${\cal CN} (0,1)$. The entries of $X_{ik}$ take values from a set such that {\small$\mathbb{E}\left[X_{ik}X^H_{ik}\right]=I_4$}. The precoders $V_{ik}$ are chosen as given below.
{
\begin{align*}
& V_{11}=E^{F'} {V^{F'}_1}, ~~ V_{12}=E^{F'} {V^{F'}_2}, \\
&V_{21}={H'}^{-1}_{22}H'_{12}V_{11}, ~ ~~V_{22} = {H'}^{-1}_{21} H'_{11} V_{12}
\end{align*}}where, {\small$E^{F'} \in \mathbb{C}^{12 \times 12}$} denotes a matrix whose columns are the eigen vectors of the matrix {\small$F'={H'}^{-1}_{11}H'_{21}{H'}^{-1}_{22}H'_{12}$}, \mbox{{\small$V^{F'}_1=I_4 \otimes[1 ~1 ~0]^T$}}, and {\small$V^{F'}_2=I_4 \otimes[1 ~0 ~1]^T$}. With the above choice of precoders, the interference symbols are aligned and (\ref{eqn-JS}) can be re-written as
\begin{align}
\nonumber
& Y'_1=\sqrt{\frac{3P}{2}}\left(H'_{11}V_{11}X_{11}+H'_{21}V_{21}X_{21}\right.\\
\nonumber
&\hspace{4cm}\left.+H_{11}V_{12}\left(X_{12}+X_{22}\right)\right)+N'_1\\
\label{eqn-JS_Aligned}
& Y'_2=\sqrt{\frac{3P}{2}}\left(H'_{12}V_{12}X_{12}+H'_{22}V_{22}X_{22}\right.\\
\nonumber
&\hspace{4cm}\left.+H_{12}V_{11}\left(X_{11}+X_{21}\right)\right)+N'_2.
\end{align}
It is proved in \cite{JaS} that the above scheme achieves a sum DoF of $\frac{16}{3}$ in the $(2,2,4)-X$ Network almost surely when the channel matrix entries take values from a continuous probability distribution.

\subsection{Review of LJJ Scheme} \label{subsec2}

In the LJJ transmission scheme for $(2,2,2)-X$ Network, every transmitter transmits two superposed Alamouti codes with appropriate precoding in three time slots, i.e., $T=3$. Each Alamouti code corresponds to the symbols meant for each receiver. The transmitted symbols are given by 

{\small
\begin{align*}
X_1&=\sqrt{\frac{3P}{4}}\left(V_{11}\underbrace{\begin{bmatrix}
                          x^1_{11} & -\overline{x^{2}_{11}} & 0\\
			  x^2_{11} &  ~\overline{x^{1}_{11}} & 0
                         \end{bmatrix}}_{X_{11}} + V_{12}\underbrace{\begin{bmatrix}
                          0 & x^1_{12} & -\overline{x^{2}_{12}}\\
			  0 & x^2_{12} &  ~\overline{x^{1}_{12}} 
                         \end{bmatrix}}_{X_{12}} \right)\\
X_2&=\sqrt{\frac{3P}{4}}\left(V_{22}\underbrace{\begin{bmatrix}
                          x^1_{21} & -\overline{x^{2}_{21}} & 0\\
			  x^2_{21} &  ~\overline{x^{1}_{21}} & 0
                         \end{bmatrix}}_{X_{21}} + V_{12}\underbrace{\begin{bmatrix}
                          0 & x^1_{22} & -\overline{x^{2}_{22}}\\
			  0 & x^2_{22} &  ~\overline{x^{1}_{22}} 
                         \end{bmatrix}}_{X_{22}} \right),
\end{align*}}where, $x^k_{ij}$ takes values from a set such that {\small$\mathbb{E}\left[\left|x^k_{ij}\right|^2\right]=1$}. The matrices $X_{ij}$, as defined above, correspond to the symbols generated by Tx-$i$ meant for Rx-$j$. The matrix entries $x^{k}_{ij}$ denote the $k^{\text{th}}$ symbol generated by Tx-$i$ for Rx-$j$. The precoders $V_{ij}$ are chosen as 

{\small
\begin{align}
\nonumber
&V_{11}=\frac{H^{-1}_{12}}{\sqrt{\text{tr}\left(H^{-1}_{12}H^{-H}_{12}\right)}}, ~ V_{12}=\frac{H^{-1}_{11}}{\sqrt{\text{tr}\left(H^{-1}_{11}H^{-H}_{11}\right)}}\\
\label{eqn-precoders}
&V_{21}=\frac{H^{-1}_{22}}{\sqrt{\text{tr}\left(H^{-1}_{22}H^{-H}_{22}\right)}}, ~ V_{22}=\frac{H^{-1}_{21}}{\sqrt{\text{tr}\left(H^{-1}_{21}H^{-H}_{21}\right)}}.
\end{align}}The coefficients in the square roots above make sure that the transmitters meet the average power constraint. Note that all the channel matrices and the precoders are $2 \times 2$ matrices. The above choice of precoders and the usage of Alamouti codes concatenated with all zero columns align the interference symbols while ensuring that the interference subspace is linearly independent of the signal subspace. We briefly describe how this happens at Rx-$1$. 
The output symbol matrix at Rx-$1$ is now given by

{\small
\begin{align*}
&Y_1=\sqrt{\frac{3P}{4}}H_{11}V_{11}X_{11} + \sqrt{\frac{3P}{4}}H_{21}V_{21}X_{21}  \\
&~~~~~~~~+\sqrt{\frac{3P}{4}}\begin{bmatrix}
                          0 & ax^1_{12} + bx^1_{22} & -a\overline{x^{2}_{12}}-b\overline{x^{2}_{22}}\\
			  0 & ax^2_{12} + bx^2_{22} &  ~a\overline{x^{1}_{12}}+b\overline{x^{1}_{22}} 
                         \end{bmatrix} + N_1
\end{align*}}where, $a=\frac{1}{\sqrt{\text{tr}\left(H^{-1}_{11}H^{-H}_{11}\right)}}$ and $b=\frac{1}{\sqrt{\text{tr}\left(H^{-1}_{22}H^{-H}_{22}\right)}}$. Let the effective channel matrices corresponding to the desired symbols from Tx-$1$ and Tx-$2$ to Rx-$1$ be denoted by $\hat{H}=H_{11}V_{11}$ and $\hat{G}=H_{21}V_{21}$ respectively. Define a $2 \times 3$ matrix $Y'$ whose first, second and third columns are given by

{\small
\begin{align}
Y'(:,1)=Y(:,1), ~Y'(:,2)=\overline{Y(:,1)}, ~Y'(:,3)=Y(:,3).
\end{align}} Similarly, define the matrix $N'_1$ obtained from $N_1$.  Denote the  $i^{\text{th}}$ rows of the $2 \times 3$ matrices $Y'_1$ and $N'_1$ by $Y'_1(i,:)$ and $N'_1(i,:)$ respectively, $i=1,2$. The processed output symbols at Rx-$1$ (i.e., $Y'_1$) can be written as

{\footnotesize
\begin{align}
\nonumber
\underbrace{\begin{bmatrix}
{Y'}_1^T(1,:)\\
{Y'}_1^T(2,:)
\end{bmatrix}}_{Y''_1}=\sqrt{\frac{3P}{4}}&\begin{bmatrix}
\hat{h}_{11} & \hat{h}_{12} & \hat{g}_{11} & \hat{g}_{12} & 0 & 0\\
\overline{\hat{h}_{12}} & -\overline{\hat{h}_{11}} & \overline{\hat{g}_{12}} & -\overline{\hat{h}_{11}} & 1 & 0\\
0 & 0 & 0 & 0 & 0 & -1\\
\hat{h}_{21} & \hat{h}_{22} & \hat{g}_{21} & \hat{g}_{22} & 0 & 0\\
\overline{\hat{h}_{22}} & -\overline{\hat{h}_{21}} & \overline{\hat{g}_{22}} & -\overline{\hat{h}_{21}} & 0 & 1\\
0 & 0 & 0 & 0 & 1 & 0
\end{bmatrix}\begin{bmatrix}
x^{1}_{11}\\
x^{2}_{11}\\
x^{1}_{21}\\
x^{2}_{21}\\
I_1\\
I_2
\end{bmatrix}\\
\label{eqn-Y''_1}
&\hspace{1cm}+\underbrace{\begin{bmatrix}
{N'}_1^T(1,:)\\
{N'}_1^T(2,:)
\end{bmatrix}}_{N''_1}
\end{align}}where, $I_1=a\overline{x^1_{12}}+b\overline{x^1_{22}}$ and $I_2=a\overline{x^2_{12}}+b\overline{x^2_{22}}$, and $\hat{h}_{ij}$ and $\hat{g}_{ij}$ denote the entries of the matrices $\hat{H}$ and $\hat{G}$ respectively. Note that, when $\hat{h}_{ij}$ and $\hat{g}_{ij}$ are non-zero, the interference symbols $I_1$ and $I_2$ are aligned in a subspace linearly independent of the signal subspace. So, pre-multiplying the matrix $Y''_1$ (defined in (\ref{eqn-Y''_1})) by the zero-forcing matrix given by
\begin{align}
F=\begin{bmatrix}
1 & 0 & 0 & 0 & 0 & 0\\
0 & 1 & 0 & 0 & 0 & -1\\
0 & 0 & 1 & 0 & 1 & 0\\
0 & 0 & 0 & 1 & 0 & 0
\end{bmatrix}
\end{align} yields

{\small
\begin{align} \label{eqn-eff_ch_mat}
 FY''_1=\sqrt{\frac{3P}{4}}\underbrace{\begin{bmatrix}
\hat{h}_{11} & \hat{h}_{12} & \hat{g}_{11} & \hat{g}_{12} \\
\overline{\hat{h}_{12}} & -\overline{\hat{h}_{11}} & \overline{\hat{g}_{12}} & -\overline{\hat{h}_{11}} \\
\hat{h}_{21} & \hat{h}_{22} & \hat{g}_{21} & \hat{g}_{22} \\
\overline{\hat{h}_{22}} & -\overline{\hat{h}_{21}} & \overline{\hat{g}_{22}} & -\overline{\hat{h}_{21}} \\
\end{bmatrix}}_{R}\begin{bmatrix}
x^{1}_{11}\\
x^{2}_{11}\\
x^{1}_{21}\\
x^{2}_{21}\\
\end{bmatrix} +FN''_1.  
\end{align}}
Now, note that decoding the symbols in (\ref{eqn-eff_ch_mat}) is similar to decoding symbols in a two user MAC with double antenna transmitters and a double antenna receiver. Hence, \cite{LJJ} makes use of the interference cancellation procedure for MAC \cite{NSC} to achieve low complexity symbol-by-symbol decoding. This procedure is described below.

Denote the sub-matrices of $R$, defined in (\ref{eqn-eff_ch_mat}), by

{\small
\begin{align}
\tilde{H}_1=\begin{bmatrix}\hat{h}_{11} & \hat{h}_{12} \\
\overline{\hat{h}_{12}} & -\overline{\hat{h}_{11}}
\end{bmatrix}, ~\tilde{G}_1=\begin{bmatrix}\hat{g}_{11} & \hat{g}_{12} \\
\overline{\hat{g}_{12}} & -\overline{\hat{g}_{11}}
\end{bmatrix}\\
\tilde{H}_2=\begin{bmatrix}\hat{h}_{21} & \hat{h}_{22} \\
\overline{\hat{h}_{22}} & -\overline{\hat{h}_{21}}
\end{bmatrix}, ~\tilde{G}_2=\begin{bmatrix}\hat{g}_{21} & \hat{g}_{22} \\
\overline{\hat{g}_{22}} & -\overline{\hat{g}_{21}}
\end{bmatrix}.
\end{align}}Denote the first two entries and the last two entries of the $4 \times 1$ vector $FY''_1$ by $\tilde{y}_1$ and $\tilde{y}_2$ respectively. Similarly, denote first two entries and the last two entries of the $4 \times 1$ vector $FN''_1$ by $\tilde{n}_1$ and $\tilde{n}_2$ respectively. Let 

{\footnotesize
\begin{align} \nonumber
 \tilde{y} &= \frac{\tilde{G}^H_1\tilde{y}_1}{\left|\left|\tilde{G}_1(1,:)\right|\right|^2}-\frac{\tilde{G}^H_2\tilde{y}_2}{\left|\left|\tilde{G}_2(1,:)\right|\right|^2}=\\
\label{eqn-define_ytilde}
&\sqrt{\frac{3P}{4}}\underbrace{\left[\frac{\tilde{G}^H_1\tilde{H}_1}{\left|\left|\tilde{G}_1(1,:)\right|\right|^2}-\frac{\tilde{G}^H_2\tilde{H}_2}{\left|\left|\tilde{G}_2(1,:)\right|\right|^2}\right]}_{\tilde{H}}\begin{bmatrix}
              x^1_{11}\\
              x^2_{11}
	   \end{bmatrix}\\
\nonumber
&\hspace{3cm}+\frac{\tilde{G}^H_1\tilde{n}_1}{\left|\left|\tilde{G}_1(1,:)\right|\right|^2}-\frac{\tilde{G}^H_2\tilde{n}_2}{\left|\left|\tilde{G}_2(1,:)\right|\right|^2}.
\end{align}}Note that the matrix $\tilde{H}$ also has an Alamouti structure and hence, $x^1_{11}$ and $x^2_{11}$ are symbol-by-symbol decodable. Similarly, $x^k_{21}$ is decoded at Rx-$1$, and $x^k_{12}$ and $x^k_{22}$ are symbol-by-symbol decodable at Rx-$2$, for $k=1,2$. The following theorem, given as Theorem $1$ in \cite{LJJ}, states the diversity gain achieved for each symbol.
\begin{theorem}\cite{LJJ}
A diversity gain of $2$ is achieved for $x^k_{ij}$, for all $i,j,k$. 
\end{theorem}

A sum DoF of $\frac{8}{3}$ is achieved in the $(2,2,2)-X$ Network with probability one if the effective channel matrix $R$ in (\ref{eqn-eff_ch_mat}) and a similar effective channel matrix at Rx-$2$ are full rank almost surely. The following theorem, given as Theorem $2$ in \cite{LJJ}, claims that matrix $R$ is almost surely full rank.
\begin{theorem}\cite{LJJ}
\label{thm2}
 ~When the entries of $H_{ij}$ are i.i.d. distributed as ${\cal CN}(0,1)$, the matrix $R$ defined in (\ref{eqn-eff_ch_mat}) is almost surely full rank.
\end{theorem}
 The proof given in \cite{LJJ} for the above theorem goes as follows. 

``The equivalent channel vectors for $x^1_{i1}$ and $x^2_{i1}$ are orthogonal, i.e., the first two columns of $R$ are orthogonal to each other and so are the last two columns of $R$. Further, the equivalent channel vectors of $x^k_{11}$ (i.e., first two columns of $R$) depend on the matrices $H_{11}$ and $H_{12}$, while those of $x^k_{21}$ (i.e., the last two columns of $R$) depend on $H_{21}$ and $H_{22}$. Almost surely, the equivalent channel vectors of each data stream are linearly independent and separable at Rx-$1$ (i.e., the matrix $R$ is full rank almost surely).''

Note that the matrix $R$ is full rank iff the subspaces spanned by the first two and the last two columns of $R$ do not intersect. We find that it is not obvious from the  facts mentioned in the proof of Theorem \ref{thm2} in \cite{LJJ} that these subspaces do not intersect almost surely. This is because the random variables in the first two columns are dependent and so are the random variables in the last two columns. So, it is not clear what distribution the determinant of $R$ follows or specifically whether it is continuously distributed or not. Further, note that the Jafar-Shamai scheme assured a sum DoF of $\frac{8}{3}$ when the entries of the channel matrices are distributed i.i.d. according to some continuous distribution and not necessarily ${\cal CN}(0,1)$. We now re-state Theorem \ref{thm2} and also provide a complete proof.

\begin{theorem}
\label{thm3}
 When the entries of $H_{ij}$ are distributed i.i.d. according to some continuous distribution, the matrix $R$ defined in (\ref{eqn-eff_ch_mat}) is almost surely full rank.
\end{theorem}
\begin{proof}
See Appendix \ref{appen_thm3}.
\end{proof}

We propose an extension of the LJJ scheme to $(2,2,4)-X$ Network in the next section.

\section{S-R STBC Based Transmission Scheme for $(2,2,4)-X$ Network} \label{sec4}
{\begin{figure*}
\begin{align} \label{eqn-SR_Code} 
    \begin{bmatrix}
        s^{1R}+js^{3I} & -s^{2R}+js^{4I} & e^{j\theta}\left(s^{5R}+js^{7I}\right) & e^{j\theta}\left(-s^{6R}+js^{8I}\right) \\
	s^{2R}+js^{4I} &  s^{1R}-js^{3I} & e^{j\theta}\left(s^{6R}+js^{8I}\right) & e^{j\theta}\left(s^{5R}-js^{7I}\right) \\
	e^{j\theta}\left(s^{7R}+js^{5I}\right) & e^{j\theta}\left(-s^{8R}+js^{6I}\right) & s^{3R}+js^{1I} & -s^{4R}+js^{2I} \\
	e^{j\theta}\left(s^{8R}+js^{6I}\right) & e^{j\theta}\left(s^{7R}-js^{5I}\right) & s^{4R}+js^{2I} &  s^{3R}-js^{1I}\\
     \end{bmatrix}
\end{align}
\hrule
\end{figure*}}
{\begin{figure*}
\begin{align} \label{eqn-SR_Code_for_X_Ch1} 
   X_{i1}= \begin{bmatrix}
        x^{1R}_{i1}+jx^{3I}_{i1} & -x^{2R}_{i1}+jx^{4I}_{i1} & 0 &e^{j\theta}\left(x^{5R}_{i1}+jx^{7I}_{i1}\right) & e^{j\theta}\left(-x^{6R}_{i1}+jx^{8I}_{i1}\right) & 0 \\
	x^{2R}_{i1}+jx^{4I}_{i1} &  x^{1R}_{i1}-jx^{3I}_{i1} & 0 & e^{j\theta}\left(x^{6R}_{i1}+jx^{8I}_{i1}\right) & e^{j\theta}\left(x^{5R}_{i1}-jx^{7I}_{i1}\right)& 0 \\
	e^{j\theta}\left(x^{7R}_{i1}+jx^{5I}_{i1}\right) & e^{j\theta}\left(-x^{8R}_{i1}+jx^{6I}_{i1}\right) & 0 & x^{3R}_{i1}+jx^{1I}_{i1} & -x^{4R}_{i1}+jx^{2I}_{i1}& 0 \\
	e^{j\theta}\left(x^{8R}_{i1}+jx^{6I}_{i1}\right) & e^{j\theta}\left(x^{7R}_{i1}-jx^{5I}_{i1}\right) & 0 & x^{4R}_{i1}+jx^{2I}_{i1} &  x^{3R}_{i1}-jx^{1I}_{i1}& 0 \\
     \end{bmatrix}
\end{align}
\hrule
\end{figure*}}
{\begin{figure*}
\begin{align} \label{eqn-SR_Code_for_X_Ch2} 
   X_{i2}= \begin{bmatrix}
        0 & x^{1R}_{i2}+jx^{3I}_{i2} & -x^{2R}_{i2}+jx^{4I}_{i2} & 0 & e^{j\theta}\left(x^{5R}_{i2}+jx^{7I}_{i2}\right) & e^{j\theta}\left(-x^{6R}_{i2}+jx^{8I}_{i2}\right)\\
	0 & x^{2R}_{i2}+jx^{4I}_{i2} &  x^{1R}_{i2}-jx^{3I}_{i2} & 0 & e^{j\theta}\left(x^{6R}_{i2}+jx^{8I}_{i2}\right) & e^{j\theta}\left(x^{5R}_{i2}-jx^{7I}_{i2}\right)\\
	0 & e^{j\theta}\left(x^{7R}_{i2}+jx^{5I}_{i2}\right) & e^{j\theta}\left(-x^{8R}_{i2}+jx^{6I}_{i2}\right)& 0 & x^{3R}_{i2}+jx^{1I}_{i2} & -x^{4R}_{i2}+jx^{2I}_{i2}\\
	0 & e^{j\theta}\left(x^{8R}_{i2}+jx^{6I}_{i2}\right) & e^{j\theta}\left(x^{7R}_{i2}-jx^{5I}_{i2}\right)& 0 & x^{4R}_{i2}+jx^{2I}_{i2} &  x^{3R}_{i2}-jx^{1I}_{i2}\\
     \end{bmatrix}
\end{align}
\hrule
\end{figure*}}

In this section, the LJJ scheme is extended to $(2,2,4)-X$ Network by exploiting a repetitive Alamouti structure (upto scaling by a constant) in the S-R STBC. This transmission scheme is proved to achieve the sum DoF of $(2,2,4)-X$ Network, and a diversity gain of at least four when fixed finite constellations are used at the transmitters. The S-R STBC proposed for $4 \times 2$ single user MIMO system in \cite{SrR1} is given by (\ref{eqn-SR_Code}) (at the top of the next page) where, $s^{i}$ denotes the $i^{th}$ complex symbol generated by the transmitter, and $\theta \in (0,2\pi)$. Note that $8$ complex symbols are transmitted in $4$ channel uses.

If $8$ complex symbols are transmitted from each transmitter to every receiver in $6$ channel uses in the $(2,2,4)-X$ Network then, a total of $\frac{16}{3}$ complex symbols per channel use is transmitted. This is done using the S-R STBC as follows. The transmitted symbols are given by

{\small
\begin{align*}
X_1&=\sqrt{\frac{3P}{4}}\left(V_{11}X_{11} + V_{12}X_{12} \right)\\
X_2&=\sqrt{\frac{3P}{4}}\left(V_{21}X_{21} + V_{22}X_{22} \right)
\end{align*}}where, the matrices $X_{i1}$ and $X_{i2}$ are given in (\ref{eqn-SR_Code_for_X_Ch1}) and (\ref{eqn-SR_Code_for_X_Ch2}) respectively, for $i=1,2$, and $x^k_{ij}$ take values from a set such that {$\mathbb{E}\left[\left|x^k_{ij}\right|^2\right]=1$}. The matrices $X_{ij}$ correspond to the symbols generated by Tx-$i$ meant for Rx-$j$. The matrix entries $x^{k}_{ij}$ denote the $k^{\text{th}}$ symbol generated by Tx-$i$ for Rx-$j$. The choice of precoders $V_{ij}$ is the same as in the LJJ scheme, i.e., given by (\ref{eqn-precoders}), where the channel matrices $H_{ij}$ are $4 \times 4$ matrices. The output symbol matrix at Rx-$1$ is given by

{\small
\begin{align*}
&Y_1=\sqrt{\frac{3P}{4}}\left(H_{11}V_{11}X_{11} + H_{21}V_{21}X_{21} \right) \\
&+\sqrt{\frac{3P}{4}}\left(\frac{1}{\sqrt{\text{tr}\left(H^{-1}_{11}H^{-H}_{11}\right)}}X_{12} + \frac{1}{\sqrt{\text{tr}\left(H^{-1}_{21}H^{-H}_{21}\right)}}X_{22} \right)+N_1
\end{align*}}where, $Y_1 \in \mathbb{C}^{4 \times 6}$. Note that the third and the sixth columns of {\small$V_{11}X_{11} + V_{21}X_{21}$} are zero. This shall be exploited for interference cancellation as follows. 

Define a matrix $Y'_1 \in \mathbb{C}^{4 \times 4}$ obtained by processing $Y_1$ as follows.

{\small
\begin{align}
\label{eqn-first_col}
& Y'_1(:,1)=Y_1(:,1),\\
\label{eqn-third_col}
& Y'_1(:,3)=Y_1(:,4),\\
\label{eqn-second_col1}
& Y'_1(1,2)=Y_1(1,2) - \overline{Y_1(2,3)},\\ 
\label{eqn-second_col2}
& Y'_1(2,2)=Y_1(2,2) + \overline{Y_1(1,3)},\\
\label{eqn-second_col3}
& Y'_1(3,2)=Y_1(3,2) - e^{j2\theta}\overline{Y_1(4,3)}, \\
\label{eqn-second_col4}
& Y'_1(4,2)=Y_1(4,2) + e^{j2\theta}\overline{Y_1(3,3)}\\
\label{eqn-fourth_col1}
& Y'_1(1,4)=Y_1(1,5) - e^{j2\theta}\overline{Y_1(2,6)},\\ 
\label{eqn-fourth_col2}
& Y'_1(2,4)=Y_1(2,5) + e^{j2\theta}\overline{Y_1(1,6)},\\
\label{eqn-fourth_col3}
& Y'_1(3,4)=Y_1(3,5) - \overline{Y_1(4,6)}, \\
\label{eqn-fourth_col4}
& Y'_1(4,4)=Y_1(4,5) + \overline{Y_1(3,6)}.
\end{align}}Note that, in (\ref{eqn-first_col}) and (\ref{eqn-third_col}), the first and the fourth columns of $Y_1$ are retained without further processing because they are interference free. These are interference free because the first and fourth columns of $X_{i2}$ are zero, for $i=1,2$. In (\ref{eqn-second_col1})-(\ref{eqn-second_col4}), the interference term associated with the second column of $Y_1$ is canceled using the third column of $Y_1$. Similarly, in (\ref{eqn-fourth_col1})-(\ref{eqn-fourth_col4}), the interference term associated with the fifth column of $Y_1$ is canceled using the sixth column of $Y_1$. Note that the conjugation and scaling of terms in the R.H.S. of (\ref{eqn-second_col1})-(\ref{eqn-fourth_col4}) involve only the third and sixth columns of $Y_1$. This interference cancellation procedure does not affect the desired symbols because the third and sixth columns of {\small$V_{11}X_{11} + V_{21}X_{21}$} are zero. Note that the LJJ scheme for $(2,2,2)-X$ Network also involves similar interference cancellation procedure though it was explained through zero-forcing of aligned interference in Section \ref{subsec2}.

Now, the matrix $Y'_1$ can be re-written as
{
\begin{align} \label{eqn-Y'_1}
 Y'_1= H_{11}V_{11}X'_{11} + H_{21}V_{21}X'_{21} + N'_1
\end{align}
}where, $X'_{i1}$ is given by (\ref{eqn-SR_Code_1}) (at the top of the next page), for $i=1,2$, and $N'_1 \in \mathbb{C}^{4 \times 4}$ is a Gaussian noise matrix whose first and third column entries are distributed as i.i.d. ${\cal CN}(0,1)$ while the second and fourth column entries are distributed as i.i.d. ${\cal CN}(0,2)$.
{\begin{figure*}
\begin{align} \label{eqn-SR_Code_1} 
   X'_{i1}= \begin{bmatrix}
        x^{1R}_{i1}+jx^{3I}_{i1} & -x^{2R}_{i1}+jx^{4I}_{i1} & e^{j\theta}\left(x^{5R}_{i1}+jx^{7I}_{i1}\right) & e^{j\theta}\left(-x^{6R}_{i1}+jx^{8I}_{i1}\right) \\
	x^{2R}_{i1}+jx^{4I}_{i1} &  x^{1R}_{i1}-jx^{3I}_{i1} &  e^{j\theta}\left(x^{6R}_{i1}+jx^{8I}_{i1}\right) & e^{j\theta}\left(x^{5R}_{i1}-jx^{7I}_{i1}\right) \\
	e^{j\theta}\left(x^{7R}_{i1}+jx^{5I}_{i1}\right) & e^{j\theta}\left(-x^{8R}_{i1}+jx^{6I}_{i1}\right) & x^{3R}_{i1}+jx^{1I}_{i1} & -x^{4R}_{i1}+jx^{2I}_{i1} \\
	e^{j\theta}\left(x^{8R}_{i1}+jx^{6I}_{i1}\right) & e^{j\theta}\left(x^{7R}_{i1}-jx^{5I}_{i1}\right)  & x^{4R}_{i1}+jx^{2I}_{i1} &  x^{3R}_{i1}-jx^{1I}_{i1} \\
     \end{bmatrix}
\end{align}
\hrule
\end{figure*}}The matrices $X'_{i2}$ is defined in a similar way as $X'_{i1}$, for $i=1,2$.

We now proceed to evaluate the diversity gain achieved by the above scheme when fixed finite constellation inputs are used at the transmitters. Towards that end, we have the following definition from \cite{ZaR}.

\begin{definition}\cite{ZaR}
  The Coordinate Product Distance (CPD) between any two signal points $u=u^R+ju^I$ and $v=v^R+jv^I$, for $u \neq v$, in a finite constellation ${\cal S}$ is defined as 
\begin{align*}
 CPD(u,v)=\left|u^R-v^R\right|\left|u^I-v^I\right|
\end{align*}and the minimum of this value among all possible pairs is defined as the CPD of ${\cal S}$.
\end{definition}

We assume that each symbol $x^{k}_{ij}$ takes values from a finite constellation whose CPD is non-zero, for all $i,j,k$. As observed in \cite{ZaR}, if a finite constellation has a zero CPD, it can always be rotated appropriately so that the resulting constellation has a non-zero CPD. Now, define the difference matrix $\triangle{X'_{ij}}^{k_1,k_2}$ by
\begin{align*}
 \triangle {X'_{ij}}^{k_1,k_2} = {X'_{ij}}^{k_1}-{X'_{ij}}^{k_2}
\end{align*}where, ${X'_{ij}}^{k_1}$ and ${X'_{ij}}^{k_2}$ denote two different realizations (i.e., $k_1 \neq k_2$) of the matrix $X'_{ij}$.

The following lemma shall be useful in establishing the diversity gain of the proposed scheme.

\begin{lemma} \label{lem3}
 There exists $\theta$ such that the difference matrix $\triangle {X'_{ij}}^{k_1,k_2}$ is full rank for all $k_1 \neq k_2$ and for all $i,j$.
\end{lemma}
\begin{proof}
 See Appendix \ref{appen_lem3}.
\end{proof}

Henceforth, we shall assume that $\theta$ is chosen so that the difference matrix {\small$\triangle {X'_{ij}}^{k_1,k_2}$} is full rank for all $k_1 \neq k_2$ and for all $i,j$.  We shall assume that ML Decoding of {\small$X'_{11}$} and {\small$X'_{21}$} is done from (\ref{eqn-Y'_1}) and ML Decoding of {\small$X'_{12}$} and {\small$X'_{22}$} is done from a similar processed received symbol matrix at Rx-$2$. The diversity gain of the proposed scheme can be obtained from the following theorem.

\begin{theorem} \label{thm4}
The average pair-wise error probability $P_e$ for the pairs of codewords {\small$\left({X'_{11}}^{k_1},{X'_{21}}^{k_2}\right)$} and {\small$\left({X'_{11}}^{k'_1},{X'_{21}}^{k'_2}\right)$} is upper bounded as

{\small\begin{align*}
 P_e\left({\left({X'_{11}}^{k_1},{X'_{21}}^{k_2}\right)\rightarrow\left({X'_{11}}^{k'_1},{X'_{21}}^{k'_2}\right)}\right) \leq c P^{-4}.
\end{align*}}for some constant $c>0$.
\end{theorem}
\begin{proof}
 See Appendix \ref{appen_thm4}.
\end{proof}

Hence, using the union bound on the average probability of error given that a particular symbol is transmitted and using Theorem \ref{thm4}, we obtain that ML decoding of {\small$X'_{11}$} and {\small$X'_{21}$} from (\ref{eqn-Y'_1}) gives a diversity gain of four. 

We shall now evaluate the DoF achievable using the proposed scheme. For the DoF evaluation we do not assume any restriction on the value of $\theta$. 

\begin{theorem} \label{thm5}
 The proposed scheme can achieve a node to node DoF of $\frac{4}{3}$ and hence, a sum DoF of $\frac{16}{3}$ with symbol-by-symbol decoding.
\end{theorem}
\begin{proof}
 See Appendix \ref{appen_thm5}.
\end{proof}

Thus, the proposed scheme achieves the sum DoF of $(2,2,4)-X$ Network using local CSIT while the Jafar-Shamai scheme requires global CSIT. 

In the following section, we shall present some simulation results comparing the probability of error performance of the proposed scheme with other schemes using finite constellation inputs.

\section{Simulation Results} \label{sec5}
In this section, we present some simulation results that include comparing the error performance of the proposed scheme for $(2,2,4)-X$ Network with that of a TDMA scheme, and the Jafar-Shamai scheme. In the TDMA scheme, the channel is used half the time by one transmitter while the other switches off. When Tx-$i$ is switched on, half the time is allocated to transmit to each of the receivers. To ensure a fair comparison, we assume TDMA with CSIT, and the symbol vectors meant to be transmitted are precoded using the full diversity precoders proposed in \cite{SrR2} for single user MIMO system with square QAM constellation inputs.

We shall briefly review the precoding technique proposed in \cite{SrR2} for single user MIMO system. We shall call the precoder as S-R Precoder. Consider a single user MIMO system with $M$ transmit and $M$ receive antennas. Full CSIT and CSIR are assumed. The channel is assumed to be quasi-static and all the channel gains are distributed as i.i.d. ${\cal CN}(0,1)$. The channel model is given by
\begin{align} \label{eqn-single_user}
 Y=\sqrt{\frac{SNR}{M}}HQX+N
\end{align}where, $Y \in \mathbb{C}^{M \times 1}$ denotes the output symbol vector, $H \in \mathbb {C}^{M \times M}$ denotes the channel matrix, $Q \in \mathbb {C}^{M \times M}$ denotes the precoder matrix, $X \in \mathbb{C}^{M \times 1}$ denotes the transmitted symbol vector, and $N \in \mathbb{C}^{M \times 1}$ denotes the Gaussian noise vector with the entries distributed as i.i.d. ${\cal CN}(0,1)$. The signal to noise ratio at each receive antenna is denoted by $SNR$ and $\mathbb E\left[X^H X\right]=M$. The transmitted symbol vector is given by $X=\left[x_1 ~x_2 \cdots x_M\right]^T$ where the symbols $x_i$ take values from a square QAM whose average power is taken to be equal to one, for $i=1,2, \cdots, M$. Let the singular value decomposition of $H$ be given by $H=U D V^H$ where, $U$ and $V$ are unitary matrices of size $M \times M$, and {\small$D= \text{diag}\left(\lambda_1(H), \lambda_2(H), \cdots, \lambda_M(H)\right)$} with {\small$\lambda_1(H) \geq \lambda_2(H) \geq \cdots \geq \lambda_M(H)$}.

The precoding matrix $Q$ is given by $Q=VP$ where, $P \in \mathbb C^{M \times M}$. Multiplying the received vector $Y$ by $U^H$ we have,
\begin{align*}
 Y'=U^H Y = \sqrt{\frac{SNR}{M}} D P X + N'
\end{align*}where, $N'=U^H N$ has the same distribution as $N$. The matrix $P$ for $M=4$ is given by 

{\footnotesize
 \begin{align*} \begin{bmatrix}
  P_1(1,1) & 0 & 0 & P_1(1,2)\\
 0 & P_2(1,1) & P_2(1,2) & 0 \\
0 & P_2(2,1) & P_2(2,2) & 0 \\
P_1(2,1) & 0 & 0 & P_1(2,2) \end{bmatrix}
 \end{align*}}where, $P_i(j.k)$ denotes the $j^{\text{th}}$ row, $k^{\text{th}}$ column element of the matrix $P_i$ given by

{\small\begin{align*}
P_i = \sqrt{2\tau_i^2} \begin{bmatrix}
                       \text{cos} ~\psi_i ~\text{cos} ~\theta_i & -\text{cos} ~\psi_i ~\text{sin} ~\theta_i\\
			\text{sin} ~\psi_i ~\text{sin} ~\theta_i & \text{sin} ~\psi_i ~\text{cos} ~\theta_i
                      \end{bmatrix} \text{, for $i=1,2$}.
\end{align*}}The values of $\tau_i$, $\psi_i$, and $\theta_i$ are selected based on the matrix $D$. The selection of values of these variables is involved and hence, the readers are referred to \cite{SrR2} for details. Similarly, for $M=2$, the matrix $P$ is given by 

{\small\begin{align*}
P   = \sqrt{2\tau_3^2} \begin{bmatrix}
                       \text{cos} ~\psi_3 ~\text{cos} \theta_3 & -\text{cos} ~\psi_3 ~\text{sin} \theta_3\\
			\text{sin} ~\psi_3 ~\text{sin} \theta_3 & \text{sin} ~\psi_3 ~\text{cos} \theta_3
                      \end{bmatrix}.
\end{align*}}Among the class of precoders having a real matrix $P$, the above choice of $P$ was shown to be approximately optimal in minimizing the ML metric given by
\begin{align} \label{eqn-single_user_dmin}
 \min_{X}\left|\left|Y'-\sqrt{\frac{SNR}{M}} D P X\right|\right|^2.
\end{align}Further, the precoders were proven to achieve full diversity.

We first compare the error probability performance of the LJJ scheme with the TDMA scheme using S-R Precoder in the $(2,2,2)-X$ Network. Such a comparison was not done in \cite{LJJ}. The value of $SNR$ in the S-R precoder is set as $2P$ to account for time sharing. In the LJJ scheme we perform ML decoding of the symbols directly from the processed receive symbol vector $F Y''_1$ given in (\ref{eqn-eff_ch_mat}) rather than symbol-by-symbol decoding as described in Section \ref{subsec2}. The transmitted symbols in the LJJ scheme are decoded using the sphere decoder \cite{ViB}. Since each transmitter achieves a rate of $\frac{4}{3}$ cspcu and $1$ cspcu in the LJJ scheme and the TDMA scheme respectively, we use $8$-QAM constellation\footnote{Here, we take $8$-QAM constellation input to be the Cartesian product of a $4$-PAM constellation that constitutes the real part and a $2$-PAM constellation that constitutes the imaginary part.} input for the LJJ scheme and $16$-QAM constellation input for the TDMA scheme using S-R Precoder so that the spectral efficiency achieved is $4$ bits/sec/Hz per transmitter. Fig. \ref{fig2} compares the Word Error Probability (WEP) of the LJJ scheme with $8$-QAM input with that of the TDMA scheme using S-R Precoder with $16$-QAM input. The TDMA scheme using S-R Precoder clearly outperforms the LJJ scheme inspite of the higher constellation size because the former has a diversity gain of $4$ while the latter has a diversity gain that is strictly greater than $2$ but lesser than $3$. Thus, the sum DoF optimality of the LJJ scheme does not translate to a better WEP performance compared to the TDMA scheme with finite constellation inputs, even at low values of $P$. 

{\begin{figure}[htbp]
\centering
\includegraphics[totalheight=3.1in,width=3.8in]{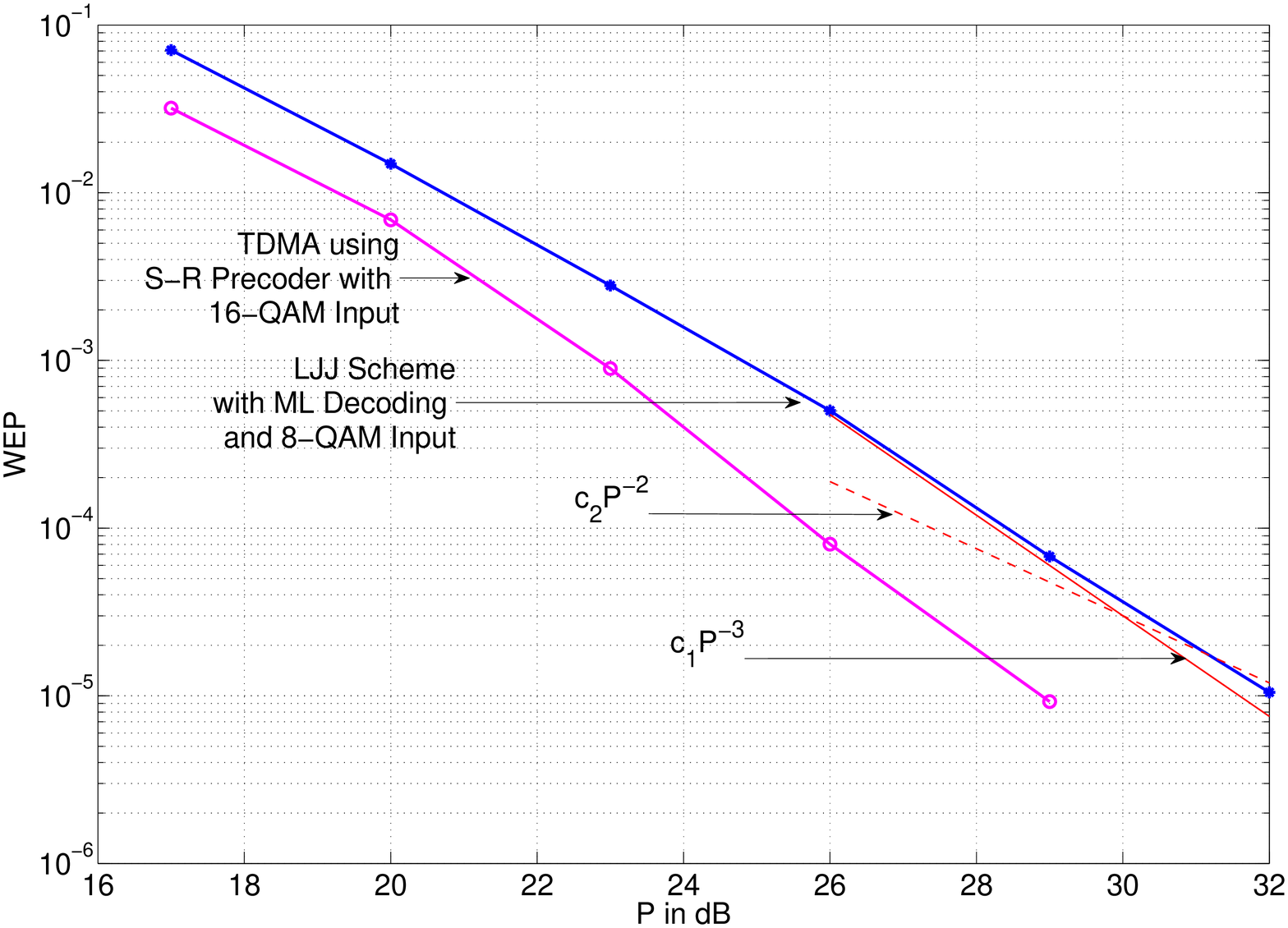}
\caption{WEP of LJJ scheme with $8$-QAM input versus WEP of TDMA using S-R Precoder with $16$-QAM input at a spectral efficiency of $4$ bits/sec/Hz per transmitter.}
 \label{fig2}
\end{figure}}

A similar result is observed with the proposed scheme for $(2,2,4)-X$ Network which we term as the modified S-R STBC scheme. Here, the TDMA scheme achieves a rate of $2$ cspcu per transmitter. Sphere decoder is used to decode the transmitted symbols from (\ref{eqn-Y'_1}) in the modified S-R STBC scheme. We simulate the TDMA scheme using S-R Precoder with $16$-QAM input and the modified S-R STBC scheme with $8$-QAM input so that the achieved spectral efficiency is $8$ bits/sec/Hz per transmitter. We have set $\theta=\frac{\pi}{4}$ in the modified S-R STBC scheme, and the constellations are rotated by an angle $\phi=\frac{\text{tan}^{-1}(2)}{2}$ to ensure a non-zero CPD \cite{ZaR}. It was shown in \cite{SrR1} that the difference matrices of the S-R STBC are full rank with  $\theta=\frac{\pi}{4}$ and $\phi=\frac{\text{tan}^{-1}(2)}{2}$ when $16$-QAM inputs are used. Since, the $8$-QAM constellation is a subset of the $16$-QAM constellation, $\triangle {X'_{ij}}^{k_1,k_2}$ is full rank for all $k_1,k_2$ and for all $i,j$. Hence, by Theorem \ref{thm4}, a diversity of four is assured for the modified S-R STBC scheme. It can be observed from Fig. \ref{fig3} that the TDMA scheme using S-R Precoder with $16$-QAM input outperforms the modified S-R STBC scheme with $8$-QAM input. Hence, like in the LJJ scheme, the sum DoF superiority of the modified S-R STBC scheme for $(2,2,4)-X$ Network over the TDMA scheme doesn't translate to superiority in terms of WEP when finite constellation inputs are used, even at low values of $P$. Note that the diversity gain offered by the TDMA scheme using S-R Precoder is $16$ whereas the modified S-R STBC scheme has an assured diversity gain of only $4$. Fig. \ref{fig3} however shows that the diversity gain offered by the modified S-R STBC scheme is strictly greater than $4$.
{\begin{figure}[htbp]
\centering
\includegraphics[totalheight=3.1in,width=3.6in]{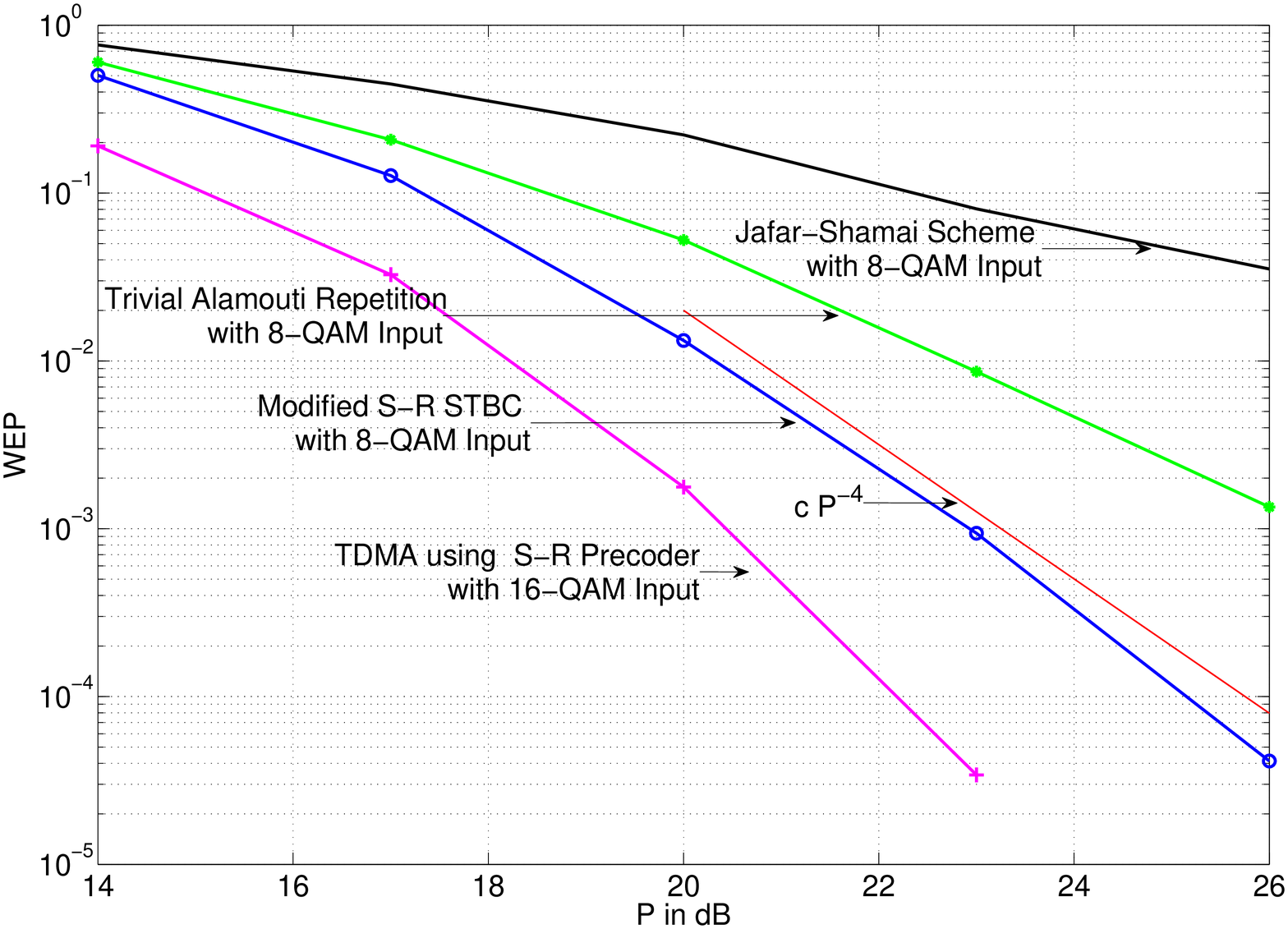}
\caption{WEP of modified S-R STBC scheme with $8$-QAM input versus WEP of TDMA using S-R Precoder with $16$-QAM at a spectral efficiency of $8$ bits/sec/Hz per transmitter.}
\label{fig3}
\end{figure}}

The precoding technique in \cite{SrR2} however applies only to square QAM constellations which can be written as a Cartesian product of two PAM constellations. Also, optimizing the precoder to minimize (\ref{eqn-single_user_dmin}) for a single user MIMO system while assuring a particular diversity gain for arbitrary constellations is an open problem. In such a scenario, there is no guarantee that TDMA with some precoding would surely outperform the LJJ scheme for $(2,2,2)-X$ Network or the modified S-R STBC scheme for $(2,2,4)-X$ Network at all values of $P$. Moreover, the TDMA scheme achieves integer rates of $1$ cspcu and $2$ cspcu per transmitter in the $(2,2,2)-X$ Network and the $(2,2,4)-X$ Network respectively whereas the LJJ scheme and the modified S-R STBC scheme achieve fractional rates of $\frac{4}{3}$ cspcu and $\frac{8}{3}$ cspcu per transmitter respectively. So, equating the spectral efficiencies for WEP comparison requires the use of higher QAM sizes than what are used in Fig. \ref{fig2} and Fig. \ref{fig3}.  Further, the decoding complexity, even with sphere decoding, is enormous for higher constellation sizes for the LJJ scheme and the modified S-R STBC scheme. Hence, it is not feasible to compare the WEP performance of the LJJ scheme and the modified S-R STBC scheme with the TDMA scheme using S-R Precoding with higher QAM sizes.

We now compare the WEP performance of the modified S-R STBC scheme with the Jafar-Shamai scheme. We shall also observe the importance of selection of $\theta$ so that  $\triangle {X'_{ij}}^{k_1,k_2}$ is full rank for all $k_1,k_2$ and for all $i,j$. Let us call the scheme that uses $\theta=0$ and $\phi=\frac{\text{tan}^{-1}(2)}{2}$ as the trivial Alamouti repetition scheme. It is easy to observe that, with the same constellation used for all the symbols and when $\theta=0$, $\triangle {X'_{ij}}^{k_1,k_2}$ is not full rank for some $k_1,k_2$, for all $i,j$. Thus, Theorem \ref{thm4} is not applicable for this case. For convenience, the scheme that uses $\theta=\frac{\pi}{4}$ and $\phi=\frac{\text{tan}^{-1}(2)}{2}$ is termed as the modified S-R STBC scheme. In the Jafar-Shamai scheme, MAP decoding of the desired symbols from (\ref{eqn-JS_Aligned}) reduces to ML decoding of all the symbols at high values of $P$ \cite{RaG}, i.e.,

{\small\begin{align*}
&(\hat{X}_{11},\hat{X}_{21}) = \arg \hspace{-0.7cm}\min_{X_{11},X_{21},X_{12}+X_{22}}\left|\left|Y'_1-\sqrt{\frac{3P}{2}}\left(H'_{11}V_{11}X_{11}\right.\right.\right.\\
&\hspace{3.2cm}\left.\left.\left.+H'_{21}V_{21}X_{21}\right)+H'_{11}V_{12}\left(X_{12}+X_{22}\right)\right|\right|^2\\
&(\hat{X}_{12},\hat{X}_{22}) = \arg \hspace{-0.7cm}\min_{X_{12},X_{22},X_{11}+X_{21}}\left|\left|Y'_2-\sqrt{\frac{3P}{2}}\left(H'_{12}V_{12}X_{12}\right.\right.\right.\\
&\hspace{3.2cm}\left.\left.\left.+H'_{22}V_{22}X_{22}\right)+H'_{12}V_{11}\left(X_{11}+X_{21}\right)\right|\right|^2.
\end{align*}}Hence, as noted in \cite{RaG} sphere decoder can be used when QAM constellations are employed. Fig. \ref{fig3} and Fig. \ref{fig4} compare the WEP of the modified S-R STBC scheme with that of the trivial Alamouti repetition scheme and the Jafar-Shamai scheme, using $8$-QAM inputs and $4$-QAM inputs respectively. It can observed from Fig. \ref{fig3} and Fig. \ref{fig4} that the modified S-R STBC scheme clearly outperforms the trivial Alamouti repetition scheme and the Jafar-Shamai scheme. 
{\begin{figure}[htbp]
\centering
\includegraphics[totalheight=3.1in,width=3.6in]{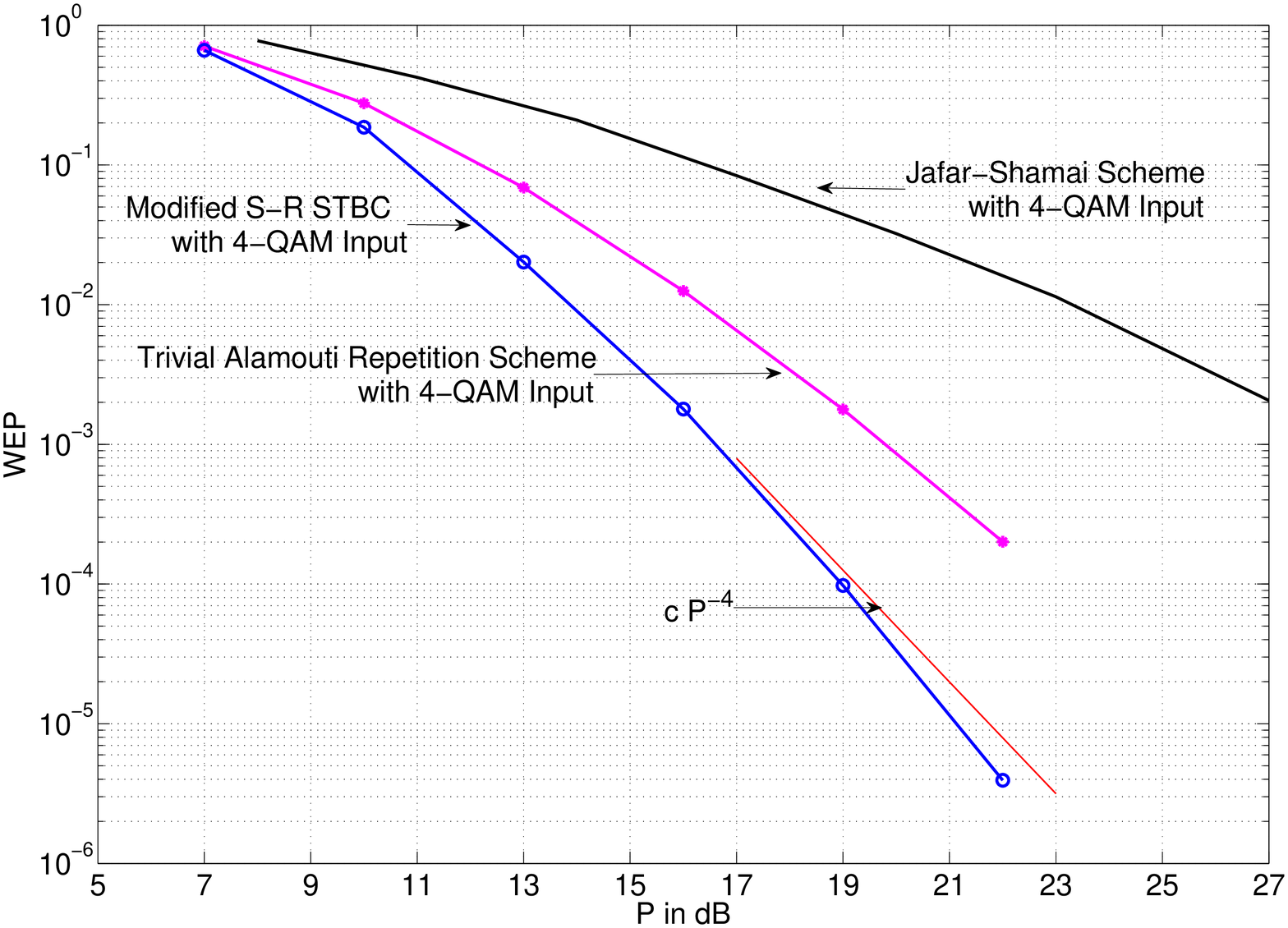}
\caption{WEP of modified S-R STBC scheme versus Trivial Alamouti Repetition and Jafar-Shamai scheme with $4$-QAM input at a spectral efficiency of $\frac{16}{3}$ bits/sec/Hz per transmitter.}
\label{fig4}
\end{figure}}

{\begin{figure}[htbp]
\centering
\includegraphics[totalheight=3.1in,width=3.6in]{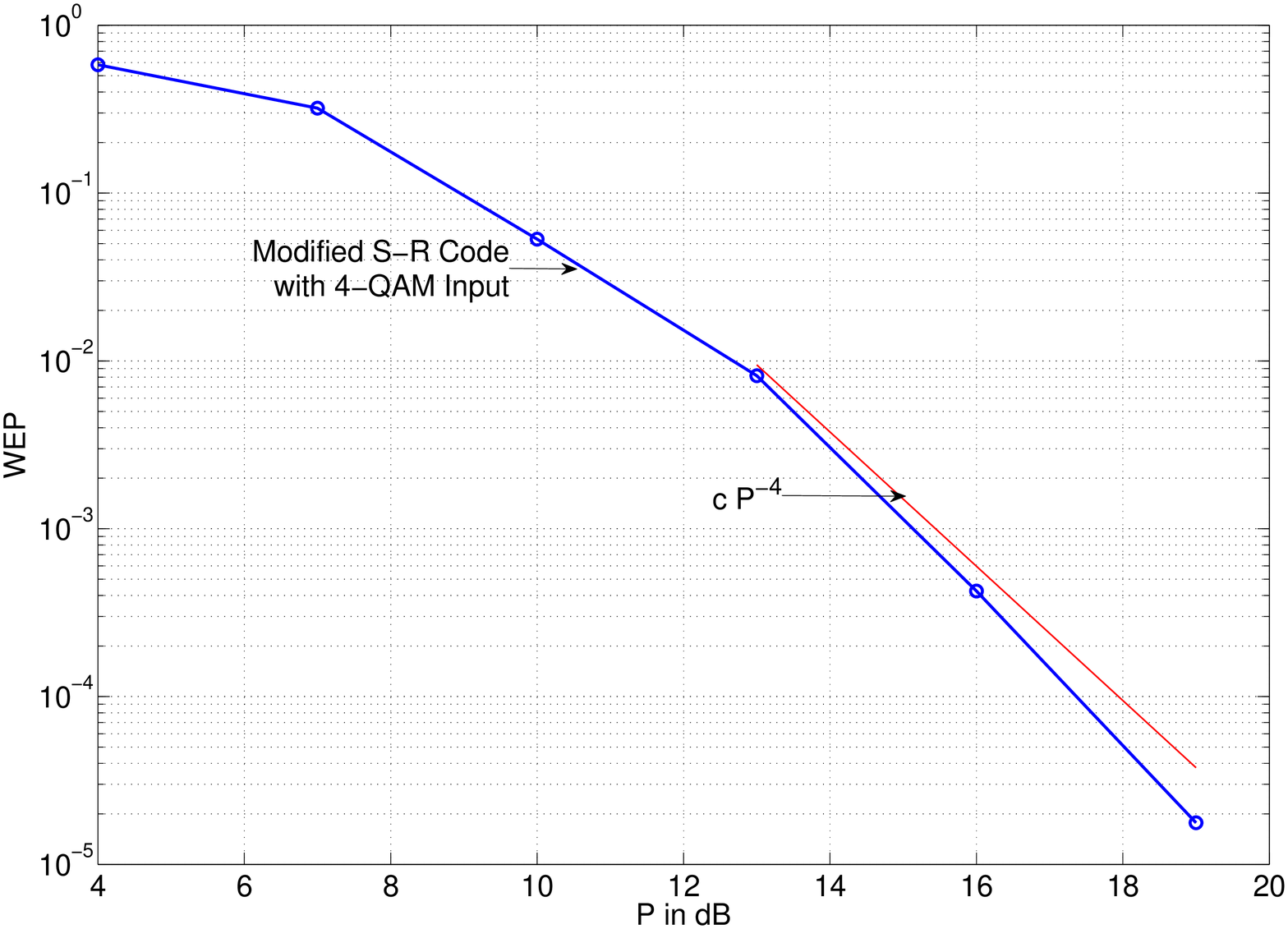}
\caption{WEP of modified S-R STBC with BPSK input  at a spectral efficiency of $\frac{8}{3}$ bits/sec/Hz per transmitter.}
\label{fig5}
\end{figure}}

In all the figures, the modified S-R scheme is found to offer a diversity gain that is strictly greater than $4$. For additional clarity, the modified S-R scheme is plotted with BPSK inputs in Fig. \ref{fig5} which also shows that the diversity gain is strictly greater than $4$. Intuitively, the modified S-R scheme achieves full receive diversity while the transmit diversity is affected because of precoding.

\section{Conclusion} \label{sec6}
A new transmission scheme based on the S-R STBC was proposed for the $(2,2,4)-X$ Network as an extension of the LJJ scheme for the $(2,2,2)-X$ Network. The proposed transmission scheme was proven to achieve the sum DoF of the $(2,2,4)-X$ Network which is equal to $\frac{16}{3}$. In comparison with the Jafar-Shamai scheme, the proposed scheme has reduced CSIT requirements. Moreover, the proposed scheme was proven to achieve a diversity gain of four when finite constellation inputs are used. Simulation results confirmed that the proposed scheme performs better in terms of error probability when compared with the Jafar Shamai scheme.

An interesting question that remains to be addressed is - what is the maximum diversity gain achievable at a sum rate of $\frac{8}{3}$ cspcu and  $\frac{16}{3}$ cspcu in the $(2,2,2)-X$ Network and $(2,2,4)-X$ Network respectively? Another interesting direction of research is to identify similar schemes for other values of $M$ so that the sum DoF of $(2,2,M)-X$ Network can be achieved with lesser CSIT requirement compared to the Jafar-Shamai scheme along with full receive diversity gain when finite constellation inputs are used.

\appendices
 \section{Proof of Theorem \ref{thm3}}
\label{appen_thm3}
\begin{IEEEproof}
We do not attempt a direct proof for showing that the matrix $R$ is full rank as the determinant expression is complicated. Instead, we shall prove it using some information theoretic inequalities and exploit the interference cancellation procedure given in (\ref{eqn-define_ytilde}). First, note that the entries of the noise vector $FN''_1$ in (\ref{eqn-eff_ch_mat}) are i.i.d. with the first and last entries being distributed as ${\cal CN}(0,1)$, and the second and third entries being distributed as ${\cal CN}(0,2)$. We now consider a modified system model where, a Gaussian noise vector $N'''_1$ is added to (\ref{eqn-eff_ch_mat}) so that the entries of the effective noise vector in (\ref{eqn-eff_ch_mat}) shall be distributed as i.i.d. ${\cal CN}(0,2)$. Henceforth in this proof, (\ref{eqn-eff_ch_mat}) is considered to be an equation with this extra noise $N'''_1$ added. The vector $\tilde{y}$ in (\ref{eqn-define_ytilde}) is also assumed to be derived from the vector in (\ref{eqn-eff_ch_mat}) with the noise $N'''_1$ added.
Define the vector $\tilde{z}$, similar to $\tilde{y}$ in (\ref{eqn-define_ytilde}), as 

 {\footnotesize
\begin{align}
\nonumber
 \tilde{z} &= \frac{\tilde{H}^H_1\tilde{y}_1}{\left|\left|\tilde{H}_1(1,:)\right|\right|^2}-\frac{\tilde{H}^H_2\tilde{y}_2}{\left|\left|\tilde{H}_2(1,:)\right|\right|^2}\\
\label{eqn-define_ztilde}
&=\sqrt{\frac{3P}{4}}\underbrace{\left[\frac{\tilde{H}^H_1\tilde{G}_1}{\left|\left|\tilde{H}_1(1,:)\right|\right|^2}-\frac{\tilde{H}^H_2\tilde{G}_2}{\left|\left|\tilde{H}_2(1,:)\right|\right|^2}\right]}_{\tilde{G}}\begin{bmatrix}
              x^1_{21}\\
              x^2_{21}
	   \end{bmatrix}\\
\nonumber
&\hspace{3cm}+\frac{\tilde{H}^H_1\tilde{n}_1}{\left|\left|\tilde{H}_1(1,:)\right|\right|^2}-\frac{\tilde{H}^H_2\tilde{n}_2}{\left|\left|\tilde{H}_2(1,:)\right|\right|^2}
\end{align}} We now have the following useful lemmas.
\begin{lemma} \label{lem1}
 The vector norms {\small$\left|\left|\tilde{G}_1(1,:)\right|\right|$} and {\small$\left|\left|\tilde{H}_1(1,:)\right|\right|$} are almost surely non-zero.
\end{lemma}
\begin{proof}
 We shall prove the statement only for {\small$\left|\left|\tilde{G}_1(1,:)\right|\right|$} and the proof for {\small$\left|\left|\tilde{H}_1(1,:)\right|\right|$} is similar. To prove this, it is sufficient to prove that $\hat{g}_{11}$ is non-zero almost surely. Note that $\hat{g}_{11}$ is given by
\begin{align*}
 \hat{g}_{11}=h_{21_{11}}v_{11_{11}}+h_{21_{12}}v_{11_{21}}.
\end{align*}Conditioned on the random matrix $V_{11}$ and the random variable $h_{21_{12}}$, if $v_{11_{11}}$ is non-zero then, $\hat{g}_{11}$ is non-zero almost surely. This is because the continuously distributed random variable $h_{21_{11}}$ is independent of $V_{11}$ and $h_{21_{12}}$, and $v_{11_{11}}$ just scales $h_{21_{11}}$ while $h_{21_{12}}v_{11_{21}}$ shifts the mean. Thus, if $v_{11_{11}}$ is almost surely non-zero then, $\hat{g}_{11}$ is also non-zero almost surely. This is explained as follows. Suppose that $v_{11_{11}}$ is zero with some non-zero probability, and consider such events. Since {\small$V_{11}=\frac{H^{-1}_{12}}{\sqrt{\text{tr}\left(H^{-1}_{12}H^{-H}_{12}\right)}}$}, we have

{\footnotesize
\begin{align} \label{eqn-g11_non_zero}
 \sqrt{\text{tr}\left(H^{-1}_{12}H^{-H}_{12}\right)} \begin{bmatrix}
                                                       h_{12_{11}} & h_{12_{12}}\\
						       h_{12_{21}} & h_{12_{22}}
                                                      \end{bmatrix}\begin{bmatrix}
                                                       0  & h^{(-1)}_{12_{12}}\\
						       h^{(-1)}_{12_{21}} & h^{(-1)}_{12_{22}}
                                                      \end{bmatrix}=\begin{bmatrix}
                                                       1  & 0\\
						       0 & 1
                                                      \end{bmatrix}
\end{align}}where, $h^{(-1)}_{12_{ij}}$ denotes the $ij^{\text{th}}$ element of $H^{-1}_{12}$. Clearly, {\small$\text{tr}\left(H^{-1}_{12}H^{-H}_{12}\right)$} is non-zero almost surely because {\small$\text{tr}\left(H^{-1}_{12}H^{-H}_{12}\right)=0$} would require all the entries of $H^{-1}_{12}$ to be equal to zero. From (\ref{eqn-g11_non_zero}) we have, 

{\small
\begin{align*}
 h_{12_{12}}h^{(-1)}_{12_{21}}=\frac{1}{\sqrt{\text{tr}\left(H^{-1}_{12}H^{-H}_{12}\right)}}, \text{ and} ~h_{12_{22}}h^{(-1)}_{12_{21}}=0.
\end{align*}}This necessitates that {\small$h_{12_{22}}=0$} as $h_{12_{12}} \neq 0$ almost surely. However, {\small$h_{12_{22}}\neq 0$} almost surely. Thus, $v_{11_{11}}$ cannot be equal to zero with non-zero probability. Hence, $\hat{g}_{11}$ is also non-zero almost surely.
\end{proof}

\begin{lemma}
\label{lem2}
 If at least one of the entries in both the matrices $\tilde{H}$ (defined in (\ref{eqn-define_ytilde})) and $\tilde{G}$ (defined in (\ref{eqn-define_ztilde})) are non-zero then, the matrix $R$ is full rank.
\end{lemma}
\begin{proof}
Note that $\tilde{H}$ and $\tilde{G}$ are Alamouti matrices. If at least one of the entries in both these matrices are non-zero then, both the matrices are full rank. Using chain rule for mutual information and data processing inequality, for any fixed value of channel matrices, we have

{\footnotesize
\begin{align}
\nonumber
& I\left[x^1_{11},x^2_{11},x^1_{21},x^2_{21};FY''_1\right]\\
\nonumber
& = I\left[x^1_{11},x^2_{11};FY''_1\right] + I\left[x^1_{21},x^2_{21};FY''_1\left|\right.x^1_{11},x^2_{11}\right]\\
\label{eqn-data_proc_ineq}
& \geq I\left[x^1_{11},x^2_{11};\tilde{y}\right] + I\left[x^1_{21},x^2_{21};\tilde{z}\left|\right.x^1_{11},x^2_{11}\right]\\
\nonumber
& = I\left[x^1_{11},x^2_{11};\tilde{y}\right] + I\left[x^1_{21},x^2_{21};\tilde{z}\right].
\end{align}}Assume that the symbols $x^1_{11},x^2_{11},x^1_{21}$, and $x^2_{21}$ are distributed as i.i.d. ${\cal CN}(0,1)$. Note that the covariance matrix of the noise vectors {\small$\frac{\tilde{G}^H_1\tilde{n}_1}{\left|\left|\tilde{G}_1(1,:)\right|\right|^2}-\frac{\tilde{G}^H_2\tilde{n}_2}{\left|\left|\tilde{G}_2(1,:)\right|\right|^2}$} and {\small$\frac{\tilde{H}^H_1\tilde{n}_1}{\left|\left|\tilde{H}_1(1,:)\right|\right|^2}-\frac{\tilde{H}^H_2\tilde{n}_2}{\left|\left|\tilde{H}_2(1,:)\right|\right|^2}$} are given by {\small$2\left(\frac{1}{\left|\left|\tilde{G}_1(1,:)\right|\right|^2}+\frac{1}{\left|\left|\tilde{G}_2(1,:)\right|\right|^2}\right)I_2$} and {\small$2\left(\frac{1}{\left|\left|\tilde{H}_1(1,:)\right|\right|^2}+\frac{1}{\left|\left|\tilde{H}_2(1,:)\right|\right|^2}\right)I_2$} respectively. From Lemma \ref{lem1}, these covariance matrices are well defined, invertible and hence, can be whitened. Now, if $\tilde{H}$ and $\tilde{G}$ are full rank then, following exactly the same steps in Section $3.2$ of \cite{Tel} we have\footnote{The effective channel matrices used while following the steps in Section $3.2$ of \cite{Tel} should be $\Sigma^{-\frac{1}{2}}_1\tilde{H}$ and $\Sigma^{-\frac{1}{2}}_2\tilde{G}$, where $\Sigma_1$ and $\Sigma_2$ are the covariance matrices of the noise vectors associated with $\tilde{H}$ and $\tilde{G}$ respectively.},
\begin{align}
\nonumber
& I\left[x^1_{11},x^2_{11};\tilde{y}\right] = 2~log(P) + o(log(P)),\text{ and} \\
\label{eqn-2logP}
&I\left[x^1_{21},x^2_{21};\tilde{z}\right]= 2~log(P)+ o(log(P)).
\end{align}
Suppose that the matrix $R$ is not full rank. Then, following the same steps in Section $3.2$ of \cite{Tel} we have, 
\begin{align}
\label{eqn-<4logP}
 I\left[x^1_{11},x^2_{11},x^1_{21},x^2_{21};FY''_1\right] =d~log(P)+ o(log(P))
\end{align}where, $d=\text{rank}(R)$ is strictly less than $4$. However, from (\ref{eqn-data_proc_ineq}) and (\ref{eqn-2logP}) we have, {\small$I\left[x^1_{11},x^2_{11},x^1_{21},x^2_{21};FY''_1\right] \geq 4~log(P)+o(log(P))$}. This contradicts (\ref{eqn-<4logP}) which states that {\small$I\left[x^1_{11},x^2_{11},x^1_{21},x^2_{21};FY''_1\right]$} grows as $d~log(P)$, where $d<4$. Hence, the matrix $R$ is full rank.
\end{proof}

Lemma \ref{lem2} states that, in order to prove Theorem \ref{thm3}, it is sufficient to show that both the matrices $\tilde{H}$ and $\tilde{G}$ contain at least one non-zero entry almost surely. We shall prove this statement only for $\tilde{H}$ and the proof for $\tilde{G}$ is similar. 

Since $\tilde{G}_1$  is an Alamouti matrix, its columns form a basis for the two dimensional vector space $\mathbb{C}^2$ over the field of complex numbers. Hence, the first column of $\tilde{H}_1$ can be written as a linear combination of the columns of $\tilde{G}_1$. The entries of the first column of $\tilde{G}^H_1 \tilde{H}_1$ are equal to the dot product of the two columns of $\tilde{G}_1$ with the first column of $\tilde{H}_1$. Hence, the first column of $\tilde{G}^H_1 \tilde{H}_1$ is a non-zero vector iff $\tilde{G}_1$ and $\tilde{H}_1$ are both non-zero matrices. From Lemma \ref{lem1}, this is true almost surely. Let \footnote{Note that the set of Alamouti matrices are closed with respect to matrix multiplication \cite{NSC}.}
\begin{align*}
 \tilde{G}^H_1 \tilde{H}_1 = \begin{bmatrix}
                              a & b\\
			      \overline{b} & -\overline{a}.
                             \end{bmatrix}
\end{align*}where, {\small$a=\overline{\hat{g}_{11}}\hat{h}_{11}+{\hat{g}_{12}}\overline{\hat{h}_{12}}$}, and {\small$b=\overline{\hat{g}_{11}}\hat{h}_{12}-{\hat{g}_{12}}\overline{\hat{h}_{11}}$}. Since the first column of $\tilde{G}^H_1 \tilde{H}_1$ is a non-zero vector almost surely, one of the following must be true almost surely: $(1)$ $a \neq 0, b=0$, $(2)$ $a = 0, b \neq 0$, or $(3)$ $a \neq 0, b \neq 0$. We now consider the case  $a \neq 0, b=0$ to prove that $\tilde{H}$ contains at least one non-zero entry almost surely.

Since $\hat{H}=H_{11}V_{11}$, we have

{\small \begin{align}
	  \nonumber
         &a=\overline{\hat{g}_{11}}\left(h_{11_{11}}v_{11_{11}}+h_{11_{12}}v_{11_{21}}\right) + \hat{g}_{12}\left(\overline{h_{11_{11}}}\overline{v_{11_{12}}}+\overline{h_{11_{12}}}\overline{v_{11_{22}}}\right)\\
	   \nonumber
	&= h^R_{11_{11}}\left(\overline{\hat{g}_{11}}v_{11_{11}}+\hat{g}_{12}\overline{v_{11_{12}}}\right) + jh^I_{11_{11}}\left(\overline{\hat{g}_{11}}v_{11_{11}}-\hat{g}_{12}\overline{v_{11_{12}}}\right) \\
	  \label{eqn-a}
        & + h^R_{11_{12}}\left(\overline{\hat{g}_{11}}v_{11_{21}}+\hat{g}_{12}\overline{v_{11_{22}}}\right) + jh^I_{11_{12}}\left(\overline{\hat{g}_{11}}v_{11_{21}}-\hat{g}_{12}\overline{v_{11_{22}}}\right).
        \end{align}}Clearly, if $a \neq 0$ then, at least one among the coefficients of {\small$h^R_{11_{11}}, h^I_{11_{11}}, h^R_{11_{12}}, h^I_{11_{12}}$} in (\ref{eqn-a}) is non-zero. Without loss of generality, consider the coefficient of $h^R_{11_{11}}$ to be non-zero. Now, let 
\begin{align*}
 \tilde{G}^H_2 \tilde{H}_2 = \begin{bmatrix}
                              c & d\\
			      \overline{d} & -\overline{c}
                             \end{bmatrix}
\end{align*}where, {\small$c=\overline{\hat{g}_{21}}\hat{h}_{21}+{\hat{g}_{22}}\overline{\hat{h}_{22}}$}, and {\small$d=\overline{\hat{g}_{21}}\hat{h}_{22}-{\hat{g}_{22}}\overline{\hat{h}_{21}}$}. Substituting for $\hat{h}_{21}$ and $\hat{h}_{22}$, $c$ can be written as

{\small\begin{align}\nonumber
        &c= h^R_{11_{21}}\left(\overline{\hat{g}_{21}}v_{11_{11}}+\hat{g}_{22}\overline{v_{11_{12}}}\right) + jh^I_{11_{21}}\left(\overline{\hat{g}_{21}}v_{11_{11}}-\hat{g}_{22}\overline{v_{11_{12}}}\right) \\
	  \label{eqn-c}
        & + h^R_{11_{22}}\left(\overline{\hat{g}_{21}}v_{11_{21}}+\hat{g}_{22}\overline{v_{11_{22}}}\right) + jh^I_{11_{22}}\left(\overline{\hat{g}_{21}}v_{11_{21}}-\hat{g}_{22}\overline{v_{11_{22}}}\right).
       \end{align}} The first row, first column entry of $\tilde{H}$ is given by {\small$\frac{a}{\left|\left|\tilde{G}_1(1,:)\right|\right|^2}-\frac{c}{\left|\left|\tilde{G}_2(1,:)\right|\right|^2}$}. Note that $a$ depends on the random variable {\small$h^R_{11_{11}}$} while $c$ depends on another independent set of random variables {\small$h^R_{11_{21}}, h^I_{11_{21}}, h^R_{11_{22}}$}, and {\small$h^I_{11_{22}}$}. Since {\small$h^R_{11_{11}}$} is continuously distributed and independent of other random variables involved in (\ref{eqn-a}) and (\ref{eqn-c}), {\small$\frac{a}{\left|\left|\tilde{G}_1(1,:)\right|\right|^2}-\frac{c}{\left|\left|\tilde{G}_2(1,:)\right|\right|^2}$} is non-zero almost surely. Hence, the first row, first column entry of $\tilde{H}$ is non-zero almost surely conditioned on the fact that $a \neq 0$. Similarly it can be proved for the other cases, i.e., $a = 0, b \neq 0$, and $a \neq 0, b \neq 0$, that at least one entry of $\tilde{H}$ is non-zero almost surely. The proof that at least one entry of $\tilde{G}$ is non-zero almost surely is similar to that for $\tilde{H}$. Thus, at least one entry of the matrices $\tilde{H}$ and $\tilde{G}$ are non-zero almost surely.  Hence, from Lemma \ref{lem2}, the matrix $R$ is also full rank.
\end{IEEEproof}

\section{Proof of Lemma \ref{lem3}}
\label{appen_lem3}
\begin{IEEEproof}
{\begin{figure*}
\begin{align} \label{eqn-prod_delta_mats}
\triangle C \triangle A^{-1} \triangle B=  e^{j2\theta} \frac{1}{|a_1|^2+|a_2|^2}\begin{bmatrix}\overline{a_1}a_3 a_5 - a_2 a_4 a_5 - a_1 a_4 \overline{a_6} - \overline{a_2} a_3 \overline{a_6}  &  -\overline{a_1} \overline{a_4} a_5 - a_2 \overline{a_3} a_5 - a_1 \overline{a_3} \overline{a_6} + \overline{a_2} \overline{a_4} \overline{a_6}\\
a_1 a_4 \overline{a_5} + \overline{a_2} a_3 \overline{a_5} + \overline{a_1} a_3 a_6 - a_2 a_4 a_6  &  a_1 \overline{a_3} \overline{a_5} - \overline{a_2} \overline{a_4} \overline{a_5} - \overline{a_1} \overline{a_4} a_6 - a_2 \overline{a_3} a_6
\end{bmatrix}
\end{align}
\hrule
\end{figure*}
}
 We shall prove the statement for $\triangle X'_{11}$\footnote{We have suppressed the superscript $k_1,k_2$ for convenience.} (i.e., $i=j=1$) and the proof for other $\triangle X_{ij}$ are similar. Define the sub-matrices of $X'_{11}$ by 

{\small\begin{align*}
&A= \begin{bmatrix}
        x^{1R}_{i1}+jx^{3I}_{i1} & -x^{2R}_{i1}+jx^{4I}_{i1}\\
	x^{2R}_{i1}+jx^{4I}_{i1} &  x^{1R}_{i1}-jx^{3I}_{i1}
	\end{bmatrix}\\
&B=\begin{bmatrix}
        e^{j\theta}\left(x^{5R}_{i1}+jx^{7I}_{i1}\right) & e^{j\theta}\left(-x^{6R}_{i1}+jx^{8I}_{i1}\right) \\
	e^{j\theta}\left(x^{6R}_{i1}+jx^{8I}_{i1}\right) & e^{j\theta}\left(x^{5R}_{i1}-jx^{7I}_{i1}\right) \\
	\end{bmatrix}\\
&C=\begin{bmatrix}
       e^{j\theta}\left(x^{7R}_{i1}+jx^{5I}_{i1}\right) & e^{j\theta}\left(-x^{8R}_{i1}+jx^{6I}_{i1}\right)\\
	e^{j\theta}\left(x^{8R}_{i1}+jx^{6I}_{i1}\right) & e^{j\theta}\left(x^{7R}_{i1}-jx^{5I}_{i1}\right)\\
	\end{bmatrix}\\
&D=\begin{bmatrix}
       x^{3R}_{i1}+jx^{1I}_{i1} & -x^{4R}_{i1}+jx^{2I}_{i1}\\
	x^{4R}_{i1}+jx^{2I}_{i1} &  x^{3R}_{i1}-jx^{1I}_{i1} \\
	\end{bmatrix}
\end{align*}}so that {\small$X'_{11}=\begin{bmatrix}
                                  A & B\\
				  C & D
                                 \end{bmatrix}$}. Now, consider the difference matrices $\triangle {X'_{11}}$ such that $\triangle A \neq \mathbf{0}$, $\triangle B \neq \mathbf{0}$, $\triangle C \neq \mathbf{0}$, and $\triangle D \neq \mathbf{0}$. The determinant of $\triangle {X'_{11}}$ can be written as

{\small \begin{align} \label{eqn-det_delta_x}
         &\left|\triangle {X'_{11}}\right|=\left|\triangle A\right| \left|\triangle D-\triangle C \triangle A^{-1} \triangle B \right|
        \end{align}}Denote the entries of $\triangle A$, $\triangle B$, $\triangle C$, and  $\triangle D$ by 
{\small \begin{align*}
         &\triangle A = \begin{bmatrix}
                         a_1 & -\overline{a_2}\\
			 a_2  & \overline{a_1}
                        \end{bmatrix},~~\triangle B = e^{j\theta}\begin{bmatrix}
                         a_3 & -\overline{a_4}\\
			 a_4  & \overline{a_3}
                        \end{bmatrix}\\
	 &\triangle C =  e^{j\theta}\begin{bmatrix}
                         a_5 & -\overline{a_6}\\
			 a_6  & \overline{a_5}
                        \end{bmatrix},~~
	 \triangle D = \begin{bmatrix}
                         a_7 & -\overline{a_8}\\
			 a_8  & \overline{a_7}
                        \end{bmatrix}.
        \end{align*}}Now, we have {\small$\triangle A^{-1} = \frac{1}{|a_1|^2+|a_2|^2}\begin{bmatrix}
                                                 \overline{a_1} & -a_2\\
						 \overline{a_2} & a_1
                                                \end{bmatrix}$}, and the product matrix {\small$\triangle C \triangle A^{-1} \triangle B$} is given by (\ref{eqn-prod_delta_mats}) (at the top of the next page). Note that the product matrix {\small$\triangle C \triangle A^{-1} \triangle B$} cannot be a zero matrix because each matrix in the product is an Alamouti matrix. 

Clearly, {\small$\left|\triangle A\right| \neq 0$}. From (\ref{eqn-det_delta_x}), for {\small$\left|\triangle {X'_{11}}\right|$} to be non-zero, there must exist $\theta$ such that {\small$\left|\triangle D-\triangle C \triangle A^{-1} \triangle B \right|$} is non-zero. We now prove the existence of such a $\theta$. Denote the elements of the product matrix {\small$\triangle C \triangle A^{-1} \triangle B$} by {\small$e^{j2\theta}\begin{bmatrix}
a & -\overline{b}\\
b & \overline{a}                                                                                                                                                                                                                                                                                                                                                                                                                                                                                                                                                                                                                                                                                                      \end{bmatrix}
$}. We now have

{\footnotesize \begin{align*}
         &\left|\triangle D-\triangle C \triangle A^{-1} \triangle B \right| = \left|\begin{bmatrix}
                                                                               a_7-e^{j2\theta}a & -\overline{a_8}+e^{j2\theta}\overline{b}\\
									       a_8-e^{j2\theta}b & \overline{a_7}-e^{j2\theta}\overline{a}
                                                                              \end{bmatrix}\right|\\
        &=|a_7|^2+ |a_8|^2-e^{j2\theta}\left(a\overline{a_7}+\overline{a}a_7+b\overline{a_8}+\overline{b}a_8\right)+ e^{j4\theta}\left(|a|^2+|b|^2\right).
        \end{align*}}The above equation is quadratic in $e^{2j\theta}$ since {\small $\triangle C \triangle A^{-1} \triangle B \neq \mathbf{0}$}. Therefore, {\small$\left|\triangle {X'_{11}}\right|$} can be equal to zero for at most two distinct values of $e^{2j\theta}$. Since there are infinite possible choices for $e^{2j\theta}$ while there are only a finite number of difference matrices, there always exists $\theta$ such that {\small$\left|\triangle {X'_{11}}^{k_1,k_2}\right| \neq 0$}, for all $k_1,k_2$.

Now,  consider the difference matrices $\triangle {X'_{11}}^{k_1,k_2}$ such that at least one among the difference sub-matrices $\triangle A$, $\triangle B$, $\triangle C$, and $\triangle D$ is a zero matrix, for $k_1 \neq k_2$. Since we assumed that each symbol $x^{k}_{11}$ takes values from finite constellations whose CPD is non-zero, $\triangle A=0$ iff $\triangle D=0$, and $\triangle B=0$ iff $\triangle C=0$  \cite{ZaR}. If $\triangle A=\triangle D=0$ then, $\triangle {X'_{11}}^{k_1,k_2}$ is full rank as $k_1 \neq k_2$ implies that $\triangle B \neq 0$, and $\triangle C \neq0$. Similarly $\triangle {X'_{11}}^{k_1,k_2}$ is full rank when $\triangle B=\triangle C=0$, for $k_1 \neq k_2$.
\end{IEEEproof}

\section{Proof of Theorem \ref{thm4}}
\label{appen_thm4}
\begin{IEEEproof}
Consider a modified system where a Gaussian noise matrix is added to (\ref{eqn-Y'_1}) so that the entries of the effective noise matrix in (\ref{eqn-Y'_1}) are distributed as i.i.d. ${\cal CN}(0,2)$. The average pair-wise error probability for this modified system is given by

{\small\begin{align} 
\nonumber
& P_e\left({\left({X'_{11}}^{k_1},{X'_{21}}^{k_2}\right)\rightarrow\left({X'_{11}}^{k'_1},{X'_{21}}^{k'_2}\right)}\right)= \\
\label{eqn-PEP1}
& \mathbb{E}\left[Q\left(P'\sqrt{\left|\left|H_{11}V_{11}\triangle X_{11} + H_{21}V_{21}\triangle X_{21}\right|\right|^2/2}\right)\right]
\end{align}}where, {\small$\triangle X_{11}={X'_{11}}^{k_1}-{X'_{11}}^{k'_1}$}, {\small$\triangle X_{21}={X'_{21}}^{k_2}-{X'_{21}}^{k'_2}$}, and $P'=\frac{3P}{4}$. Note that either {\small$\triangle X_{11} \neq 0, \triangle X_{21}=0$} or {\small$\triangle X_{11} = 0, \triangle X_{21}\neq 0$} or {\small$\triangle X_{11} \neq 0, \triangle X_{21}\neq 0$}. We shall prove the statement of the theorem only for the case {\small$\triangle X_{11} \neq 0$}, and the proof for the rest of the cases are similar. The Frobenius norm in (\ref{eqn-PEP1}) can be re-written as

{\small \begin{align}
\nonumber
&  \left|\left|H_{11}V_{11}\triangle X_{11} + H_{21}V_{21}\triangle X_{21}\right|\right|^2 =\\
\nonumber
& \left[ \underbrace{\left(\triangle X_{11}^T V_{11}^T \otimes I_4 \right)vec(H_{11})+\left(\triangle X_{21}^T V_{21}^T \otimes I_4 \right)vec(H_{21})}_{H'}\right]^H \times\\
\label{eqn-H'}
& \left[ \left(\triangle X_{11}^T V_{11}^T \otimes I_4 \right)vec(H_{11})+\left(\triangle X_{21}^T V_{21}^T \otimes I_4 \right)vec(H_{21})\right].
\end{align}}Note that, conditioned on $H_{12}$ and $H_{22}$, the vector {\small$H'$} defined in (\ref{eqn-H'}) is a Gaussian vector with mean zero and covariance matrix $K$ given by 

{\small \begin{align}  
\label{eqn-K'}
&K=\\
\nonumber
&\left(\underbrace{\left(\triangle X_{11}^T V_{11}^T\right)\left(\triangle X_{11}^T V_{11}^T\right)^H\hspace{-0.2cm}+\hspace{-0.1cm}\left(\triangle X_{21}^T V_{21}^T\right)\left(\triangle X_{21}^T V_{21}^T\right)^H}_{K'}\right) \otimes I_4.
\end{align}}{\begin{figure*}
 \footnotesize\begin{align}
\label{eqn-div1}
&\mathbb{E}\left[Q\left(\sqrt{P'\left|\left|H_{11}V_{11}\triangle X_{11} + H_{21}V_{21}\triangle X_{21}\right|\right|^2/2}\right)\right]=\mathbb{E}_{H_{12},H_{22}}\left[\mathbb{E}_{H_{11},H_{21}|{H_{12},H_{22}}}\left[Q\left(\sqrt{P'\left|\left|H_{11}V_{11}\triangle X_{11} + H_{21}V_{21}\triangle X_{21}\right|\right|^2/2}\right)\right]\right]\\
\label{eqn-div2}
&=\mathbb{E}_{H_{12},H_{22}}\left[\mathbb{E}_{H''|{H_{12},H_{22}}}\left[Q\left(\sqrt{P'\frac{H''^H H''}{2}}\right)\right]\right]=\mathbb{E}_{H_{12},H_{22}}\left[\mathbb{E}_{H_1,H_2,H_3,H_4|{H_{12},H_{22}}}\left[Q\left(\sqrt{P'\frac{\sum_{i=1}^{4}\left|\left|K'^{\frac{1}{2}}H_i\right|\right|^2}{2}}\right)\right]\right]\\
\label{eqn-div4}
&=\mathbb{E}_{H_{12},H_{22}}\left[\mathbb{E}_{H_1,H_2,H_3,H_4|{H_{12},H_{22}}}\left[Q\left(\sqrt{P'\frac{\sum_{i=1}^{4}\text{tr}\left(H_i^H {K'^{\frac{1}{2}}}^H K'^{\frac{1}{2}} H_i \right)}{2}}\right)\right]\right]\\
\label{eqn-div5}
&=\mathbb{E}_{H_{12},H_{22}}\left[\mathbb{E}_{H_1,H_2,H_3,H_4|{H_{12},H_{22}}}\left[Q\left(\sqrt{P'\frac{\sum_{i=1}^{4}\text{tr}\left(H_i^H \Lambda H_i \right)}{2}}\right)\right]\right]\\
\label{eqn-div8}
&=\mathbb{E}_{H_{12},H_{22}}\left[\mathbb{E}_{H'_1,H'_2,H'_3,H'_4|{H_{12},H_{22}}}\left[Q\left(\sqrt{P'\frac{\sum_{i=1}^{4}\sum_{j=1}^{4} \lambda_j(K') |H_i(j)|^2 }{2}}\right)\right]\right]\\
\label{eqn-div9}
&\leq\mathbb{E}_{H_{12},H_{22}}\left[\mathbb{E}_{H'_1,H'_2,H'_3,H'_4|{H_{12},H_{22}}}\left[Q\left(\sqrt{P'\frac{\sum_{i=1}^{4}\sum_{j=1}^{4} \lambda_j(K'_1) |H_i(j)|^2 }{2}}\right)\right]\right]\\
\label{eqn-div10}
&=\mathbb{E}_{H_{12}}\left[\mathbb{E}_{H_1,H_2,H_3,H_4|{H_{12}}}\left[Q\left(\sqrt{P'\frac{\sum_{i=1}^{4}\text{tr}\left(H_i^H \Lambda_1 H_i \right)}{2}}\right)\right]\right]\\
\label{eqn-div11}
&=\mathbb{E}_{H_{12}}\left[\mathbb{E}_{H_1,H_2,H_3,H_4|{H_{12}}}\left[Q\left(\sqrt{P'\frac{\sum_{i=1}^{4}\text{tr}\left(\left(V_1^HH_i\right)^H {\Lambda_1} \left(V_1^H H_i\right) \right)}{2}}\right)\right]\right]\\
\label{eqn-div12}
&=\mathbb{E}_{H_{12}}\left[\mathbb{E}_{H_1,H_2,H_3,H_4|{H_{12}}}\left[Q\left(\sqrt{P'\frac{\sum_{i=1}^{4}\text{tr}\left(H_i^H \left(V_1 {\Lambda_1}^{\frac{1}{2}}U_1^H\right) ~\left(U_1 {\Lambda_1}^{\frac{1}{2}} V_1^H\right) H_i \right)}{2}}\right)\right]\right]\\
\label{eqn-div13}
&=\mathbb{E}_{H_{12}}\left[\mathbb{E}_{H_1,H_2,H_3,H_4|{H_{12}}}\left[Q\left(\sqrt{P'\frac{\sum_{i=1}^{4}\left|\left|{K'_1}^{\frac{1}{2}}H_i\right|\right|^2}{2}}\right)\right]\right]\\
\label{eqn-div14}
&=\mathbb{E}_{H_{12}}\left[\mathbb{E}_{H_1,H_2,H_3,H_4|{H_{12}}}\left[Q\left(\sqrt{P'\frac{\sum_{i=1}^{4} H_i^H{V^T_{11}}^H \left(\triangle X_{11} \triangle X_{11}^H\right)^TV^T_{11}H_i}{2}}\right)\right]\right]\\
\label{eqn-div15}
&=\mathbb{E}_{H_{12}}\left[\mathbb{E}_{H_1,H_2,H_3,H_4|{H_{12}}}\left[Q\left(\sqrt{P'\frac{\sum_{i=1}^{4} H_i^H\left(\left(V_{11}U_{\triangle X_{11}}\right)^T\right)^H \Lambda_{\triangle X_{11}} \left(V_{11}U_{\triangle X_{11}}\right)^T H_i}{2}}\right)\right]\right]\\
\label{eqn-div16}
&\leq \mathbb{E}_{H_{12}}\left[\mathbb{E}_{H_1,H_2,H_3,H_4|{H_{12}}}\left[Q\left(\sqrt{P' \lambda_4\left(\triangle X_{11}\right)\frac{\sum_{i=1}^{4} H_i^H\left(\left(V_{11}U_{\triangle X_{11}}\right)^T\right)^H \left(V_{11}U_{\triangle X_{11}}\right)^T H_i}{2}}\right)\right]\right]\\
\label{eqn-div17}
&= \mathbb{E}_{H_{12}}\left[\mathbb{E}_{H_1,H_2,H_3,H_4|{H_{12}}}\left[Q\left(\sqrt{P' \lambda_4\left(\triangle X_{11}\right)\frac{\sum_{i=1}^{4} H_i^H\left(V^T_{11}\right)^H V^T_{11} H_i}{2}}\right)\right]\right]\\
\label{eqn-div18}
&= \mathbb{E}_{H_{12}}\left[\mathbb{E}_{H_1,H_2,H_3,H_4|{H_{12}}}\left[Q\left(\sqrt{P' \lambda_4\left(\triangle X_{11}\right)\frac{\sum_{i=1}^{4} H_i^H U_{V_{11}} \Lambda_{V_{11}} U^H_{V_{11}} H_i}{2}}\right)\right]\right]\\
\label{eqn-div19}
&= \mathbb{E}_{H_{12}}\left[\mathbb{E}_{H_1,H_2,H_3,H_4|{H_{12}}}\left[Q\left(\sqrt{P' \lambda_4\left(\triangle X_{11}\right)\frac{\sum_{i=1}^{4} H_i^H U_{V_{11}} \Lambda_{V_{11}} U^H_{V_{11}} H_i}{2}}\right)\right]\right]\\
\label{eqn-div20}
&= \mathbb{E}_{H_{12}}\left[\mathbb{E}_{H_1,H_2,H_3,H_4|{H_{12}}}\left[Q\left(\sqrt{P' \lambda_4\left(\triangle X_{11}\right)\frac{\sum_{i=1}^{4} \left(U_{V_{11}}^H H_i\right)^H  \Lambda_{V_{11}} U^H_{V_{11}} H_i}{2}}\right)\right]\right]\\
\label{eqn-div21}
&\stackrel{(a)}{\leq} \mathbb{E}_{H_{12}}\left[\frac{1}{\prod_{j=1}^4\left(1+\frac{3P\lambda_4(\triangle X_{11})\lambda_j(V_{11})}{8}\right)^4}\right] \stackrel{(b)}{<} \frac{1}{\left(1+\frac{3P\lambda_4(\triangle X_{11})}{32}\right)^4} \stackrel{(c)}{\approx} cP^{-4}
\end{align}
\hrule
\end{figure*}}In other words, when the successive elements of {\small$H'$} are grouped in blocks of four entries each, the blocks are distributed i.i.d. as Gaussian matrix with zero mean and covariance matrix given by $K'$ which is defined in the R.H.S of (\ref{eqn-K'}). Since $K'$ is a positive semi-definite Hermitian matrix, let the eigen decomposition of the matrix $K'$ be given by $K'=U \Lambda U^H$ where, $U$ is a $4 \times 4$ unitary matrix formed by the eigen vectors of $K'$, and {\small$\Lambda=\text{diag}\left(\lambda_1(K'),\lambda_2(K'),\lambda_3(K'),\lambda_4(K')\right)$} denotes the matrix whose diagonal entries are ordered eigen values of $K'$ with {\small$\lambda_1(K')\geq\lambda_2(K')\geq\lambda_3(K')\geq\lambda_4(K')\geq 0$}. Denote a square-root matrix of {\small$K'$} by {\small$K'^{\frac{1}{2}}$}, i.e., {\small$K'=K'^{\frac{1}{2}}{K'^{\frac{1}{2}}}^H$} where, {\small$K'^{\frac{1}{2}}=U \Lambda^{\frac{1}{2}}$}. The vector $H'$ is now statistically equivalent to the following vector

{\small\begin{align*}
 H'' = \begin{bmatrix}
       K'^{\frac{1}{2}}H_1\\
        K'^{\frac{1}{2}}H_2\\
	 K'^{\frac{1}{2}}H_3\\
	   K'^{\frac{1}{2}}H_4\\
      \end{bmatrix}
\end{align*}}where, $H_i \in \mathbb{C}^{4 \times 1}$, $i=1,2,3,4$, are Gaussian vectors whose entries are distributed as i.i.d. ${\cal CN}(0,1)$. Now, (\ref{eqn-PEP1}) can be successively re-written as in (\ref{eqn-div1})-(\ref{eqn-div10}) (given at the top of the next page) where, (\ref{eqn-div2}) follows from the statistical equivalence between $H'$ and $H''$, (\ref{eqn-div4}) follows from the fact that $||A||^2= \text{tr}(A^HA)$, and (\ref{eqn-div5}) follows from the definition of $K'^{\frac{1}{2}}$. Now, define {\small$K'_1=\left(\triangle X_{11}^T V_{11}^T\right)\left(\triangle X_{11}^T V_{11}^T\right)^H$} and {\small$K'_2=\left(\triangle X_{21}^T V_{21}^T\right)\left(\triangle X_{21}^T V_{21}^T\right)^H$} so that {\small$K'=K'_1+K'_2$}. Let $\lambda_j(K'_1)$ denote the eigen values of $K'_1$ in non-increasing order from $j=1$ to $j=4$ . Using Weyl's inequalities \footnote{Weyl's inequalities relate the eigen values of sum two of Hermitian matrices with the eigen values of the individual matrices.} (see Section III.2, pp. $62$ of \cite{Bha}), we have {\small$\lambda_j(K'_1) \leq \lambda_j(K')$}, $j=1,2,3,4$. Thus, we have the inequality (\ref{eqn-div9}) from (\ref{eqn-div8}) where, $H_i(j)$ denotes the $j^{\text{th}}$ entry of the vector $H_i$. Let $K'_1=U_1 \Lambda_1 U_1^H$ denote the eigen decomposition of $K'_1$ where, {\small$\Lambda_1=\text{diag}(\lambda_1(K'_1),\lambda_1(K'_2),\lambda_1(K'_3),\lambda_1(K'_4))$}, and $U_1$ is a unitary matrix composed of eigen vectors of $K'_1$. Equation (\ref{eqn-div10}) follows from the fact that the argument inside the Q-function in (\ref{eqn-div9}) is independent of $H_{22}$. Let the singular value decomposition of $\triangle X^T_{11} V^T_{11}$ be given by $\triangle X^T_{11} V^T_{11}=U_1 \Lambda_1^{\frac{1}{2}} V^H_1$. Note that $\triangle X^T_{11} V^T_{11}$ is a square root matrix of $K'_1$ and hence, we shall denote this by ${K'_1}^{\frac{1}{2}}$. Now, (\ref{eqn-div11}) follows from the fact that the distribution of $H'_i$ is invariant to multiplication by the unitary matrix $V_1$, and using straight-forward simplifications we obtain (\ref{eqn-div14}). Now, let the eigen decomposition of {\small$\triangle X_{11}\triangle X^H_{11}$} be given by {\small$\triangle X_{11}\triangle X^H_{11}=U_{\triangle X_{11}} \Lambda_{\triangle X_{11}} U^H_{\triangle X_{11}}$} where, {\small$\Lambda_{\triangle X_{11}}$} denotes the eigen value matrix whose eigen values in non-increasing order are given by $\lambda_j\left(\triangle X_{11}\right)$, $j=1,2,3,4$. Note that $\lambda_4\left(\triangle X_{11}\right)>0$ as $\theta$ was chosen such that $\triangle X_{11}$ is full rank. Now, substituting this eigen decomposition in (\ref{eqn-div14}) we have (\ref{eqn-div15}). The inequality (\ref{eqn-div16}) follows from the fact that $\lambda_4\left(\triangle X_{11}\right)$ is the minimum eigen value of $\triangle X_{11}$, and (\ref{eqn-div17}) follows from {\small$V_{11}$} being equal to {\small$\frac{H^{-1}_{12}}{\sqrt{\text{tr}\left(H^{-1}_{12}H^{-H}_{12}\right)}}$} and the fact that the distribution of $V_{11}$ is invariant to multiplication by the unitary matrix $U_{\triangle X_{11}}$ (because $H_{12}$ is Gaussian distributed). Using the eigen decomposition of $\left(V^T_{11}\right)^H V_{11}=U_{V_{11}} \Lambda_{V_{11}} U_{V_{11}}$ and some straight-forward techniques involved in evaluating diversity as in \cite{TSC}, we obtain (\ref{eqn-div21})$(a)$. Now, note that the eigen values of $V_{11}$ are given by 
\begin{align*}
\lambda_j\left(V_{11}\right)=\frac{\frac{1}{\lambda_{5-j}\left(H_{11}\right)}}{\sum_{j=1}^{4}\frac{1}{\lambda_j\left(H_{11}\right)}} 
\end{align*}where, $\lambda_j\left(H_{11}\right)$ denote the eigen values of $H_{11}$ in non-increasing order from $j=1$ to $j=4$. Thus, $\lambda_j\left(V_{11}\right)$ can be lower bounded as
\begin{align*}
 &\lambda_j\left(V_{11}\right) \geq \frac{\frac{1}{\lambda_{5-j}\left(H_{11}\right)}}{\sum_{j=1}^{4}\frac{1}{\lambda_4\left(H_{11}\right)}} = \frac{\lambda_4\left(H_{11}\right)}{4\lambda_{5-j}\left(H_{11}\right)}.
\end{align*}For $j=1$, the above lowerbound is equal to $\frac{1}{4}$, and for $j=2,3,4$ the above lowerbound is in turn trivially lowerbounded by $0$. Hence, we obtain the inequality in (\ref{eqn-div21})$(b)$, and the approximation in (\ref{eqn-div21})$(c)$ holds good at high values of $P$, where the constant $c=\frac{32^4}{3^4\lambda^4_4\left(\triangle X_{11}\right)}$.
\end{IEEEproof}

\section{Proof of Theorem \ref{thm5}}
\label{appen_thm5}
\begin{IEEEproof}
 We shall employ an interference cancellation procedure similar to that used in the LJJ scheme in Section \ref{subsec2} to achieve symbol-by-symbol decoding. The symbols $x^k_{ij}$ are assumed to be distributed as i.i.d. ${\cal CN}(0,1)$. We now need to decode $X'_{11}$ and $X'_{21}$ from (\ref{eqn-Y'_1}) with symbol-by-symbol decoding.
We shall decode the first two and the last two columns of $X'_{i1}$ independently. 

Consider a modified system where a Gaussian noise matrix $N''_1$ is added to (\ref{eqn-Y'_1}) so that the entries of the effective noise matrix in (\ref{eqn-Y'_1}) are distributed as i.i.d. ${\cal CN}(0,2)$. The matrix $Y'_1$ defined in (\ref{eqn-Y'_1}) is now taken to be a matrix with the noise $N''_1$ added. Denote the effective channel matrices from Tx-$1$ and Tx-$2$ to Rx-$1$ by $\hat{H}=H_{11}V_{11}$ and $\hat{G}=H_{21}V_{21}$ respectively. Define the matrices $\tilde{H}$ and $\tilde{G}$ by
{\begin{align} \nonumber
 &\tilde{H} = \sqrt{\text{tr}\left(H^{-1}_{12} H^{-H}_{12}\right)}~\hat{H} = H_{11}H^{-1}_{12}\\
\label{eqn-define-H_tilde}
&\tilde{G} = \sqrt{\text{tr}\left(H^{-1}_{22} H^{-H}_{22}\right)}~\hat{G} = H_{21}H^{-1}_{22}.
\end{align}}Define a processed received symbol matrix $Y''_1 \in \mathbb{C}^{4 \times 4}$ by
{
\begin{align*} 
&Y''_1(:,1) = Y'_1(:,1), ~Y''_1(:,2) = \overline{Y'_1(:,2)}\\
&Y''_1(:,3) = Y'_1(:,3), ~Y''_1(:,4) = Y'_1(:,4).
\end{align*}}Now, the first two columns of $Y''_1$ can be re-written as

{\small\begin{align} 
\label{eqn-mod_sys_model}
\begin{bmatrix}
y''_{1_{11}}\\
y''_{1_{12}}\\ 
y''_{1_{21}}\\
y''_{1_{22}}\\
y''_{1_{31}}\\
y''_{1_{32}}\\
y''_{1_{41}}\\
y''_{1_{42}}
\end{bmatrix}
= \begin{bmatrix}        
        H_1 & H_5 & G_1 & G_5\\
	H_2 & H_6 & G_2 & G_6\\
	H_3 & H_7 & G_3 & G_7\\
	H_4 & H_8 & G_4 & G_8
       \end{bmatrix}
\begin{bmatrix}
x'^{1}_{11}\\
x'^{2}_{11}\\
x'^{7}_{11}\\
x'^{8}_{11}\\
x'^{1}_{21}\\
x'^{2}_{21}\\
x'^{7}_{21}\\
x'^{8}_{21}\\
\end{bmatrix}+N'''_1
\end{align}}where, $H_i$ and $G_i$ are defined in (\ref{eqn-define_H_i}) (at the top of the next page), for $i=1,2,\cdots,8$, and $N'''_1\in \mathbb{C}^{8 \times 1}$ is a Gaussian vector whose entries are distributed as i.i.d. ${\cal CN}(0,2)$.
{\begin{figure*} \small
\begin{align}
\nonumber
& H_1=\begin{bmatrix}
     \hat{h}_{11} & \hat{h}_{12}\\
    \overline{\hat{h}_{12}} & -\overline{\hat{h}_{11}}
    \end{bmatrix}, ~ H_2=\begin{bmatrix}
     \hat{h}_{21} & \hat{h}_{22}\\
    \overline{\hat{h}_{22}} & -\overline{\hat{h}_{21}}
    \end{bmatrix}, ~H_3=\begin{bmatrix}
     \hat{h}_{31} & \hat{h}_{32}\\
    \overline{\hat{h}_{32}} & -\overline{\hat{h}_{31}}
    \end{bmatrix}, ~H_4=\begin{bmatrix}
     \hat{h}_{41} & \hat{h}_{42}\\
    \overline{\hat{h}_{42}} & -\overline{\hat{h}_{41}}
    \end{bmatrix},\\
\label{eqn-define_H_i}
& H_5=\begin{bmatrix}
     e^{j\theta}\hat{h}_{13} & e^{j\theta}\hat{h}_{14}\\
    e^{-j\theta}\overline{\hat{h}_{14}} & -e^{-j\theta}\overline{\hat{h}_{13}}
    \end{bmatrix}, ~H_6=\begin{bmatrix}
     e^{j\theta}\hat{h}_{23} & e^{j\theta}\hat{h}_{24}\\
     e^{-j\theta}\overline{\hat{h}_{24}} & -e^{-j\theta}\overline{\hat{h}_{23}}
    \end{bmatrix}, ~H_7=\begin{bmatrix}
     e^{j\theta}\hat{h}_{33} & e^{j\theta}\hat{h}_{34}\\
     e^{-j\theta}\overline{\hat{h}_{34}} & -e^{-j\theta}\overline{\hat{h}_{33}}
    \end{bmatrix}, ~H_8=\begin{bmatrix}
     e^{j\theta}\hat{h}_{43} & e^{j\theta}\hat{h}_{44}\\
     e^{-j\theta}\overline{\hat{h}_{44}} & -e^{-j\theta}\overline{\hat{h}_{43}}
    \end{bmatrix},\\
\nonumber
& G_1=\begin{bmatrix}
     \hat{g}_{11} & \hat{g}_{12}\\
    \overline{\hat{g}_{12}} & -\overline{\hat{g}_{11}}
    \end{bmatrix}, ~ G_2=\begin{bmatrix}
     \hat{g}_{21} & \hat{g}_{22}\\
    \overline{\hat{g}_{22}} & -\overline{\hat{g}_{21}}
    \end{bmatrix}, ~G_3=\begin{bmatrix}
     \hat{g}_{31} & \hat{g}_{32}\\
    \overline{\hat{g}_{32}} & -\overline{\hat{g}_{31}}
    \end{bmatrix}, ~G_4=\begin{bmatrix}
     \hat{g}_{41} & \hat{g}_{42}\\
    \overline{\hat{g}_{42}} & -\overline{\hat{g}_{41}}
    \end{bmatrix},\\
\nonumber
& G_5=\begin{bmatrix}
     e^{j\theta}\hat{g}_{13} & e^{j\theta}\hat{g}_{14}\\
     e^{-j\theta}\overline{\hat{g}_{14}} & -e^{-j\theta}\overline{\hat{g}_{13}}
    \end{bmatrix}, ~G_6=\begin{bmatrix}
     e^{j\theta}\hat{g}_{23} & e^{j\theta}\hat{g}_{24}\\
     e^{-j\theta}\overline{\hat{g}_{24}} & -e^{-j\theta}\overline{\hat{g}_{23}}
    \end{bmatrix}, ~G_7=\begin{bmatrix}
     e^{j\theta}\hat{g}_{33} & e^{j\theta}\hat{g}_{34}\\
     e^{-j\theta}\overline{\hat{g}_{34}} & -e^{-j\theta}\overline{\hat{g}_{33}}
    \end{bmatrix}, ~G_8=\begin{bmatrix}
     e^{j\theta}\hat{g}_{43} & e^{j\theta}\hat{g}_{44}\\
     e^{-j\theta}\overline{\hat{g}_{44}} & -e^{-j\theta}\overline{\hat{g}_{43}}
    \end{bmatrix}.
\end{align}
\hrule
\end{figure*}}The symbols $x'^{k}_{ij}$ are defined in (\ref{eqn-define-x'}).
{\begin{figure*}
\small\begin{align} \nonumber
x'^{1}_{ij}=x^{1R}_{ij}+jx^{3I}_{ij}, ~~x'^{2}_{ij}=x^{2R}_{ij}+jx^{4I}_{ij}, ~~x'^{3}_{ij}=  x^{3R}_{ij}+jx^{1I}_{ij}, ~~x'^{4}_{ij} = -x^{4R}_{ij}+jx^{2I}_{ij}\\
\label{eqn-define-x'}
x'^{5}_{ij}= x^{7R}_{ij}+jx^{5I}_{ij}, ~~x'^{6}_{ij}=-x^{8R}_{ij}+jx^{6I}_{ij}, ~~x'^{7}_{ij}=x^{5R}_{ij}+jx^{7I}_{ij}, ~~x'^{8}_{ij}=-x^{6R}_{ij}+jx^{8I}_{ij}.
\end{align}
\hrule
\end{figure*}}Considering the last two columns of $Y''_1$, an equation similar to (\ref{eqn-mod_sys_model}) involving the symbols $x'^{k}_{i1}$ can be written, for $k=3,4,5,6$ and $i=1,2$. We however avoid it for the sake of brevity. We now proceed to prove that $x'^{1}_{i1}, x'^{2}_{i1}, x'^{7}_{i1}$, and $x'^{8}_{i1}$ can be recovered using interference cancellation as follows.

Let {\small$z_i=\begin{bmatrix}
          y''_{1_{i1}}\\
	  y''_{1_{i2}}
         \end{bmatrix}$}, for $i=1,2,3,4$. The interference cancellation is performed in three steps.

\textit{Step $1$:} Define the symbols obtained by eliminating the symbols $x'^{1}_{21}$ and $x'^{2}_{21}$ from (\ref{eqn-mod_sys_model}) by

{\small\begin{align}
\nonumber
&z'_1=\frac{G^H_2 z_2}{||G_2(1,:)||^2}-\frac{G^H_1 z_1}{||G_1(1,:)||^2} \\
\label{eqn-define_z'}
&z'_2=\frac{G^H_3 z_3}{||G_3(1,:)||^2}-\frac{G^H_1 z_1}{||G_1(1,:)||^2}\\
\nonumber
&z'_3=\frac{G^H_4 z_4}{||G_4(1,:)||^2}-\frac{G^H_1 z_1}{||G_1(1,:)||^2}.
\end{align}}The symbols $z'_1$, $z'_2$, and $z'_3$ can be written as

{\small\begin{align} 
\label{eqn-mod_sys_model_z'}
\begin{bmatrix}
z'_1\\
z'_2\\
z'_3
\end{bmatrix}
= \begin{bmatrix}        
        H'_1 & H'_4 & G'_1\\
	H'_2 & H'_5 & G'_2\\
	H'_3 & H_6 & G'_3
  \end{bmatrix}
\begin{bmatrix}
x'^{1}_{11}\\
x'^{2}_{11}\\
x'^{7}_{11}\\
x'^{8}_{11}\\
x'^{7}_{21}\\
x'^{8}_{21}\\
\end{bmatrix}+W'_1
\end{align}}where, the Alamouti matrices $H'_i\in \mathbb{C}^{2 \times 2}$, for $i=1,2,\cdots 6$, $G'_i\in \mathbb{C}^{2 \times 2}$, for $i=1,2,3$, are defined in (\ref{eqn-define_H'_i_and_G'_i}), and $W'_1\in \mathbb{C}^{6 \times 1}$ denotes the relevant Gaussian noise matrix.
{\begin{figure*} \small
\begin{align}  
\nonumber
&H'_1=\frac{G^H_2 H_2}{||G_2(1,:)||^2}-\frac{G^H_1 H_1}{||G_1(1,:)||^2}, ~H'_2=\frac{G^H_3 H_3}{||G_3(1,:)||^2}-\frac{G^H_1 H_1}{||G_1(1,:)||^2}, ~ H'_3=\frac{G^H_4 H_4}{||G_4(1,:)||^2}-\frac{G^H_1 H_1}{||G_1(1,:)||^2}\\
\label{eqn-define_H'_i_and_G'_i}
&H'_4=\frac{G^H_2 H_6}{||G_2(1,:)||^2}-\frac{G^H_1 H_5}{||G_1(1,:)||^2}, ~H'_5=\frac{G^H_3 H_7}{||G_3(1,:)||^2}-\frac{G^H_1 H_5}{||G_1(1,:)||^2}, ~ H'_6=\frac{G^H_4 H_8}{||G_4(1,:)||^2}-\frac{G^H_1 H_5}{||G_1(1,:)||^2}\\
\nonumber
&G'_1=\frac{G^H_2 G_6}{||G_2(1,:)||^2}-\frac{G^H_1 G_5}{||G_1(1,:)||^2}, ~G'_2=\frac{G^H_3 G_7}{||G_3(1,:)||^2}-\frac{G^H_1 G_5}{||G_1(1,:)||^2}, ~ G'_3=\frac{G^H_4 G_8}{||G_4(1,:)||^2}-\frac{G^H_1 G_5}{||G_1(1,:)||^2}.
\end{align} \hrule
\end{figure*}}

\textit{Step $2$:} Define the signals obtained by eliminating the symbols $x'^{7}_{21}$ and $x'^{8}_{21}$ from $z'_i$ (defined in (\ref{eqn-define_z'})) by

{\small\begin{align}
\nonumber
&z''_1=\frac{G'^H_2 z'_2}{||G'_2(1,:)||^2}-\frac{G'^H_1 z'_1}{||G'_1(1,:)||^2} \\
\label{eqn-define_z''}
&z''_2=\frac{G'^H_3 z'_3}{||G'_3(1,:)||^2}-\frac{G'^H_1 z'_1}{||G'_1(1,:)||^2}.
\end{align}}The symbols $z''_1$, $z''_2$, and $z''_3$ can be written as

{\small\begin{align} 
\label{eqn-mod_sys_model_z''}
\begin{bmatrix}
z''_1\\
z''_2
\end{bmatrix}
= \begin{bmatrix}        
        H''_1 & H''_3\\
	H''_2 & H''_4
  \end{bmatrix}
\begin{bmatrix}
x'^{1}_{11}\\
x'^{2}_{11}\\
x'^{7}_{11}\\
x'^{8}_{11}
\end{bmatrix}+W''_1
\end{align}}where, the Alamouti matrices $H''_i$, for $i=1,2,3,4$, are defined in (\ref{eqn-define_H''_i}), and $W''_1 \in \mathbb{C}^{4 \times 1}$ denotes the relevant Gaussian noise matrix.
{\begin{figure*} \footnotesize
\begin{align}  
\label{eqn-define_H''_i}
&H''_1=\frac{G'^H_2 H'_2}{||G'_2(1,:)||^2}-\frac{G'^H_1 H'_1}{||G'_1(1,:)||^2}, ~H''_2=\frac{G'^H_3 H'_3}{||G'_3(1,:)||^2}-\frac{G'^H_1 H'_1}{||G'_1(1,:)||^2}, ~ H''_3=\frac{G'^H_2 H'_5}{||G'_2(1,:)||^2}-\frac{G'^H_1 H'_4}{||G'_1(1,:)||^2}, ~H''_4=\frac{G'^H_3 H'_6}{||G'_3(1,:)||^2}-\frac{G'^H_1 H'_4}{||G'_1(1,:)||^2}.
\end{align} \hrule
\end{figure*}}

\textit{Step $3$:} Finally, define the signals obtained by eliminating the symbols $x'^{7}_{11}$ and $x'^{8}_{11}$ from $z''_i$ (defined in (\ref{eqn-define_z''})) by

{\small\begin{align}
\nonumber
&z'''_1=\frac{H''^H_3 z'_1}{||H''_3(1,:)||^2}-\frac{H''^H_4 z''_2}{||H''_4(1,:)||^2}\\
\label{eqn-define_z'''}
&=\left[\frac{H''^H_3 H'_1}{||H''_3(1,:)||^2}-\frac{H''^H_4 H''_2}{||H''_4(1,:)||^2}\right]\begin{bmatrix}
x'^{1}_{11}\\
x'^{2}_{11}\end{bmatrix}+W'''_1
\end{align}}where, $W'''_1\in \mathbb{C}^{2 \times 1}$ denotes the relevant Gaussian noise matrix. 

A similar interference cancellation algorithm involving the symbols $x^k_{11}$ and $x^k_{21}$, for $k=3,4,5,6$, can be written starting from the last two columns of $Y''_1$. The proof for decoding these symbols with vanishing probability of error (with respect to the codeword length)  is similar to that for $x^k_{11}$ and $x^k_{21}$, for $k=1,2,7,8$, and hence, we avoid the details. To prove that the proposed scheme achieves a node-to-node DoF of $\frac{4}{3}$ almost surely, it is sufficient to prove that at least one of the first column entries of the Alamouti matrix {\small$\left[\frac{H''^H_3 H''_1}{||H''_3(1,:)||^2}-\frac{H''^H_4 H''_2}{||H''_4(1,:)||^2}\right]$} is non-zero almost surely. This is because if {\small$\left[\frac{H''^H_3 H''_1}{||H''_3(1,:)||^2}-\frac{H''^H_4 H''_2}{||H''_4(1,:)||^2}\right]$} is a non-zero Alamouti matrix then, at least one among the matrices $H''_{4}$ or $H''_{3}$ is a non-zero Alamouti matrix. Hence, if {\small$\begin{bmatrix}
x'^{1}_{11}\\
x'^{2}_{11}                                                                                                                                                                                                                                                                                                                                                                                                                                                                                                                                                                                                                                                                                 \end{bmatrix}$} can be decoded with vanishing probability of error then clearly, from (\ref{eqn-mod_sys_model_z''}), {\small$\begin{bmatrix}x'^{7}_{11}\\
x'^{8}_{11}\end{bmatrix}$} can also be decoded with vanishing probability of error. We shall now prove that the first row, first column entry of {\small$\left[\frac{H''^H_3 H''_1}{||H''_3(1,:)||^2}-\frac{H''^H_4 H''_2}{||H''_4(1,:)||^2}\right]$} is non-zero almost surely.

Substituting for $H'_i$ in (\ref{eqn-define_H''_i}), the matrices $H''_i$ can be written as in (\ref{eqn-rewrite_H''}).
{\begin{figure*}\small
\begin{align}
\nonumber
&H''_3 = \underbrace{\frac{G'^H_2}{||G'_2(1,:)||^2}\frac{G^H_3}{||G_3(1,:)||^2}}_{E_1}H_7 - \underbrace{\frac{G'^H_1}{||G'_1(1,:)||^2}\frac{G^H_2}{||G_2(1,:)||^2}}_{E_2}H_6 + \underbrace{\left(\frac{G'^H_1}{||G'_1(1,:)||^2}- \frac{G'^H_2}{||G'_2(1,:)||^2}\right) \frac{G^H_1}{||G_1(1,:)||^2}}_{E_3} H_5,\\
\nonumber
&H''_1 = \frac{G'^H_2}{||G'_2(1,:)||^2}\frac{G^H_3}{||G_3(1,:)||^2}H_3 - \frac{G'^H_1}{||G'_1(1,:)||^2}\frac{G^H_2}{||G_2(1,:)||^2}H_2 + \left(\frac{G'^H_1}{||G'_1(1,:)||^2}- \frac{G'^H_2}{||G'_2(1,:)||^2}\right) \frac{G^H_1}{||G_1(1,:)||^2} H_1,\\
\label{eqn-rewrite_H''}
&H''_4 = \underbrace{\frac{G'^H_3}{||G'_3(1,:)||^2}\frac{G^H_4}{||G_4(1,:)||^2}}_{F_1}H_8 - \underbrace{\frac{G'^H_1}{||G'_1(1,:)||^2}\frac{G^H_2}{||G_2(1,:)||^2}}_{F_2}H_6 + \underbrace{\left(\frac{G'^H_1}{||G'_1(1,:)||^2}- \frac{G'^H_3}{||G'_3(1,:)||^2}\right) \frac{G^H_1}{||G_1(1,:)||^2}}_{F_3} H_5,\\
\nonumber
&H''_2 = \frac{G'^H_3}{||G'_3(1,:)||^2}\frac{G^H_4}{||G_4(1,:)||^2}H_4 - \frac{G'^H_1}{||G'_1(1,:)||^2}\frac{G^H_2}{||G_2(1,:)||^2}H_2 + \left(\frac{G'^H_1}{||G'_1(1,:)||^2}- \frac{G'^H_3}{||G'_3(1,:)||^2}\right) \frac{G^H_1}{||G_1(1,:)||^2} H_1.
\end{align}
\hrule
\end{figure*}}Define the matrices $E_i \in \mathbb{C}^{2 \times 2}$ and $F_i  \in \mathbb{C}^{2 \times 2}$ as in (\ref{eqn-rewrite_H''}). Denote the entries of the matrices $E_i$ by

{\small\begin{align*} 
 E_1=\begin{bmatrix}
      e_1 & e_2\\
      \overline{e_2} &-\overline{e_1}
     \end{bmatrix}, ~E_2=\begin{bmatrix}
      e_3 & e_4\\
      \overline{e_4} &-\overline{e_3}
     \end{bmatrix}E_3=\begin{bmatrix}
      e_5 & e_6\\
      \overline{e_6} &-\overline{e_5}
     \end{bmatrix}.
\end{align*}}Similarly, define the entries of the matrices $F_i$, $i=1,2,3$. 
Note that the matrices $H''_3$ and $H''_1$ depend on $\hat{h}_{3j}$ through the matrices $H_3$ and $H_7$ whereas $H''_4$ and $H''_2$ do not depend on $\hat{h}_{3j}$, for $j=1,2,3,4$. This crucial observation shall be exploited to show that the first row, first column entry of the matrix {\small$\left[\frac{H''^H_3 H''_1}{||H''_3(1,:)||^2}-\frac{H''^H_4 H''_2}{||H''_4(1,:)||^2}\right]$} is non-zero. The first row, first column entries of {\small$H''^H_3 H''_1$} and {\small$H''^H_4 H''_3$} are given in (\ref{eqn-1row_1col_IC1}) and (\ref{eqn-1row_1col_IC2}) respectively.
{\begin{figure*} \scriptsize
\begin{align} \nonumber
&\left[H''^H_3 H''_1\right]_{11}= \left(e_1e^{-j\theta}\overline{\hat{h}_{33}}+\overline{e_2}e^{j\theta}\hat{h}_{34}+e_3e^{-j\theta}\overline{\hat{h}_{23}}+\overline{e_4}e^{j\theta}\hat{h}_{24}+e_5 e^{-j\theta}\overline{\hat{h}_{13}}+\overline{e_6}e^{j\theta}\hat{h}_{14}\right)\left(\overline{e_1}\hat{h}_{31}+e_2 \overline{\hat{h}_{32}}+\overline{e_3}\hat{h}_{21}+e_{4}\overline{\hat{h}_{22}}+\overline{e_5} \hat{h}_{11}+e_6 \overline{\hat{h}_{12}}\right)\\
\label{eqn-1row_1col_IC1} 
&\hspace{1.5cm}+\left(e_2e^{-j\theta}\overline{\hat{h}_{33}}-\overline{e_1}e^{j\theta}\hat{h}_{34}+e_4e^{-j\theta}\overline{\hat{h}_{23}}-\overline{e_3}e^{j\theta}\hat{h}_{24}+e_6 e^{-j\theta}\overline{\hat{h}_{13}}-\overline{e_5}e^{j\theta}\hat{h}_{14}\right)\left(\overline{e_2}\hat{h}_{31}-e_1 \overline{\hat{h}_{32}}+\overline{e_4}\hat{h}_{21}-e_{3}\overline{\hat{h}_{22}}+\overline{e_6} \hat{h}_{11}-e_5 \overline{\hat{h}_{12}}\right)\\
\nonumber
&\left[H''^H_4 H''_2\right]_{11}= \left(f_1e^{-j\theta}\overline{\hat{h}_{43}}+\overline{f_2}e^{j\theta}\hat{h}_{44}+f_3e^{-j\theta}\overline{\hat{h}_{23}}+\overline{f_4}e^{j\theta}\hat{h}_{24}+f_5 e^{-j\theta}\overline{\hat{h}_{13}}+\overline{f_6}e^{j\theta}\hat{h}_{14}\right)\left(\overline{f_1}\hat{h}_{41}+f_2 \overline{\hat{h}_{42}}+\overline{f_3}\hat{h}_{21}+f_{4}\overline{\hat{h}_{22}}+\overline{f_5} \hat{h}_{11}+f_6 \overline{\hat{h}_{12}}\right)\\
\label{eqn-1row_1col_IC2} 
&\hspace{1.5cm}+\left(f_2e^{-j\theta}\overline{\hat{h}_{43}}-\overline{f_1}e^{j\theta}\hat{h}_{44}+f_4 e^{-j\theta}\overline{\hat{h}_{23}}-\overline{f_3}e^{j\theta}\hat{h}_{24}+f_6 e^{-j\theta}\overline{\hat{h}_{13}}-\overline{f_5}e^{j\theta}\hat{h}_{14}\right)\left(\overline{f_2}\hat{h}_{41}-f_1 \overline{\hat{h}_{42}}+\overline{f_4}\hat{h}_{21}-f_{3}\overline{\hat{h}_{22}}+\overline{f_6} \hat{h}_{11}-f_5 \overline{\hat{h}_{12}}\right).
\end{align} \hrule
\end{figure*}}Since {\small$\hat{H}=H_{11}V_{11}$}, the entries of $\hat{H}$ are given by {\small$\hat{h}_{ij}=\sum_{k=1}^{4}h_{11_{ik}}v_{11_{kj}}$}, for $i,j=1,2,3,4$. Conditioning on all the random variables except $h_{11_{31}}$ and substituting for $\hat{h}_{ij}$ in (\ref{eqn-1row_1col_IC1}) we have (\ref{eqn-final1}) which is re-written as (\ref{eqn-final2}), where $c_i$ are functions of the conditioned random variables.
{\begin{figure*}\scriptsize
\begin{align} \nonumber
& \left[H''^H_3 H''_1\right]_{11}=\left(e_1e^{-j\theta}\overline{v_{11_{13}}}\overline{h_{11_{31}}}+\overline{e_2}e^{j\theta}v_{11_{14}}h_{11_{31}}+c_1\right)\left(\overline{e_1}v_{11_{11}}h_{11_{31}}+e_2 \overline{v_{11_{12}}}\overline{h_{11_{31}}}+c_2\right)\\
\label{eqn-final1}
&\hspace{2cm}+\left(e_2 e^{-j\theta}\overline{v_{11_{13}}}\overline{h_{11_{31}}}-\overline{e_1}e^{j\theta}v_{11_{14}}h_{11_{31}}+c_3\right)\left(\overline{e_2}v_{11_{11}}h_{11_{31}}-e_1 \overline{v_{11_{12}}}\overline{h_{11_{31}}}+c_4\right)\\
\nonumber
&=\left( h^R_{11_{31}}\left[e_1e^{-j\theta}\overline{v_{11_{13}}}+\overline{e_2}e^{j\theta}v_{11_{14}}\right]+jh^I_{11_{31}}\left[-e_1 e^{-j\theta}\overline{v_{11_{13}}}+\overline{e_2}e^{j\theta}\overline{v_{11_{14}}}\right]+c_1\right)\left(h^R_{11_{31}}\left[\overline{e_1}v_{11_{11}}+e_2 \overline{v_{11_{12}}}\right]+jh^I_{11_{31}}\left[\overline{e_1}v_{11_{11}}-e_2 \overline{v_{11_{12}}}\right]+c_2\right)\\
\label{eqn-final2}
&+\left(h^R_{11_{31}}\left[e_2 e^{-j\theta}\overline{v_{11_{13}}}-\overline{e_1}e^{j\theta}v_{11_{14}}\right]+jh^I_{11_{31}}\left[-e_2 e^{-j\theta}\overline{v_{11_{13}}}-\overline{e_1}e^{j\theta}v_{11_{14}}\right]+c_3\right)\left(h^R_{11_{31}}\left[\overline{e_2}v_{11_{11}}-e_1 \overline{v_{11_{12}}}\right]+jh^I_{11_{31}}\left[\overline{e_2}v_{11_{11}}+e_1 \overline{v_{11_{12}}}\right]+c_4\right)
\end{align}\hrule
\end{figure*}}Note that the expression of {\small$\left[{H''^H_4 H''_2}\right]_{11}$} in (\ref{eqn-1row_1col_IC2}) and {\small$||H''_4(1,:)||^2$} are independent of $h_{11_{3j}}$, for all $j$. Now, the coefficients of $h^{ {\scriptscriptstyle R}^2}_{11_{31}}$ and $h^{ {\scriptscriptstyle I}^2}_{11_{31}}$ in (\ref{eqn-final2}) are given by $p$ and $-p$ respectively where,
\begin{align} \label{eqn-define_p}
& p=e^{i\theta}\left[\left(|e_1|^2+|e_2|^2\right)\overline{v_{11_{12}}}v_{11_{14}}\right]\\
\nonumber
&\hspace{2cm}+e^{-i\theta}\left[\left(|e_1|^2+|e_2|^2\right)v_{11_{11}}\overline{v_{11_{13}}}\right].
\end{align}If $p$ is non-zero then, clearly {\small$H''_3$} is a non-zero Alamouti matrix and hence, {\small$||H''_3(1,:)||^2$} is also non-zero. We now have the following useful lemmas.
\begin{lemma} \label{lem4}
At least one among $e_1$ and $e_2$ (considered now as random variables) are non-zero almost surely.
\end{lemma}
\begin{proof}
 It is easy to prove that $\frac{G^H_3}{||G_3(1,:)||^2}$ is a non-zero Alamouti matrix almost surely\footnote{The proof for this is on the same lines as that of Lemma \ref{lem1} given in Appendix \ref{appen_thm3}.}. Since {\small$E=\frac{G'^H_2}{||G'_2(1,:)||^2}\frac{G^H_3}{||G_3(1,:)||^2}$} is a product of Alamouti matrices, it is now sufficient to prove that {\small$G'_2$} is a non-zero matrix almost surely. Substituting for $G_1$, $G_3$, $G_5$, and $G_7$ from (\ref{eqn-define_H_i}) in the definition of $G'_2$, we have

{\small
\begin{align} \nonumber
&g'_{2_{11}}=\frac{1}{\left(|\hat{g}_{31}|^2+|\hat{g}_{32}|^2\right)\left(|\hat{g}_{11}|^2+|\hat{g}_{12}|^2\right)}\times\\\label{eqn-final3}&	\left(\left(|\hat{g}_{11}|^2+|\hat{g}_{12}|^2\right)		\begin{bmatrix}
                                                                                                 e^{j\theta}\overline{\hat{g}_{31}}\hat{g}_{33}+e^{-j\theta}\hat{g}_{32}\overline{\hat{g}_{34}} 
                                                                                                \end{bmatrix}\right.\\\nonumber&\hspace{1.5cm}\left.-\left(|\hat{g}_{31}|^2+|\hat{g}_{32}|^2\right)\begin{bmatrix}
                                                                                                 e^{j\theta}\overline{\hat{g}_{11}}\hat{g}_{13}+e^{-j\theta}\hat{g}_{12}\overline{\hat{g}_{14}} 
                                                                                                \end{bmatrix}\right).
\end{align}}
Note that the term outside the parenthesis in (\ref{eqn-final3}), i.e., {\small $\frac{1}{\left(|\hat{g}_{31}|^2+|\hat{g}_{32}|^2\right))\left(|\hat{g}_{11}|^2+|\hat{g}_{12}|^2\right)}$}  is non-zero almost surely. We shall now prove that the term inside the parenthesis in (\ref{eqn-final3}) is also non-zero almost surely. Since {\small$\hat{G}=H_{21}V_{21}$}, the entries $\hat{g}_{3j}$ and $\hat{g}_{1j}$ are given by {\small$\hat{g}_{3j}=\sum_{k=1}^{4}h_{21_{3k}}v_{21_{kj}}$} and  {\small$\hat{g}_{1j}=\sum_{k=1}^{4}h_{21_{1k}}v_{21_{kj}}$} respectively, for $j=1,2,3,4$. Conditioning on all the random variables except $h_{21_{31}}$, we have  {\small$\hat{g}_{3j}=h_{21_{31}}v_{21_{1j}}+q_j$} where, $q_j$ is some function of the conditioned random variables. Note that $\hat{g}_{1j}$, for all $j$, are independent of 
$h_{21_{31}}$. Considering the terms inside the parenthesis in (\ref{eqn-final3}), the coefficient of $|h_{21_{31}}|^2$ is given by (\ref{eqn-final4}) (at the top of the next page).
{\begin{figure*}
\scriptsize \begin{align} \label{eqn-final4}
&\left(|\hat{g}_{11}|^2+|\hat{g}_{12}|^2\right)e^{j\theta}\overline{v_{21_{11}}} v_{21_{13}} +  \left(|\hat{g}_{11}|^2+|\hat{g}_{12}|^2\right)e^{-j\theta}v_{21_{12}}\overline{v_{21_{14}}} -\left(|v_{21_{11}}|^2+|v_{21_{12}}|^2\right)\left(e^{j\theta}\overline{\hat{g}_{11}}\hat{g}_{13}+e^{-j\theta}\hat{g}_{12}\overline{\hat{g}_{14}}\right)\\
 \label{eqn-final4a}
&=\left(\sqrt{\text{tr}\left(H^{-1}_{22}H^{-H}_{22}\right)}\right)^4\left[\left(|\tilde{g}_{11}|^2+|\tilde{g}_{12}|^2\right)e^{j\theta}\overline{h^{(-1)}_{22_{11}}} h^{(-1)}_{22_{13}} +  \left(|\tilde{g}_{11}|^2+|\tilde{g}_{12}|^2\right)e^{-j\theta}h^{(-1)}_{22_{12}}\overline{h^{(-1)}_{22_{14}}} -\left(|h^{(-1)}_{22_{11}}|^2+|h^{(-1)}_{22_{12}}|^2\right)\left(e^{j\theta}\overline{\tilde{g}_{11}}\tilde{g}_{13}+e^{-j\theta}\tilde{g}_{12}\overline{\tilde{g}_{14}}\right)\right]
\end{align}\hrule
\end{figure*}}If this coefficient is non-zero then, further conditioning on $h^I_{21_{31}}$, the terms inside the parenthesis in (\ref{eqn-final3}) constitute a non-zero polynomial of degree $2$ in $h^R_{21_{31}}$. Since $h^R_{21_{31}}$ is continuously distributed, the term inside the parenthesis in (\ref{eqn-final3}) is almost surely non-zero. 

Hence, the proof shall be complete if we prove that the expression in (\ref{eqn-final4}) is non-zero almost surely. Substituting for $v_{21_{ij}}$, we have (\ref{eqn-final4a}) where, $h^{(-1)}_{22_{ij}}$ denotes the entries of $H^{-1}_{22}$. Since {\small$\tilde{g}_{ij}=\sum_{k=1}^{4}h_{21_{1k}}h^{(-1)}_{22_{kj}}$}, the coefficient of\footnote{The coefficient of $|h_{21_{11}}|^2$ is equal to zero. So, we consider the coefficient of $|h_{21_{12}}|^2$.}  $|h_{21_{12}}|^2$ in the term inside the parenthesis of (\ref{eqn-final4a}) is given by

{\small \begin{align} \label{eqn-final5}
&\left(\left|h^{(-1)}_{22_{21}}\right|^2+\left|h^{(-1)}_{22_{22}}\right|^2\right)\left(e^{j\theta}\overline{h^{(-1)}_{22_{11}}}h^{(-1)}_{22_{13}}+e^{-j\theta}h^{(-1)}_{22_{12}}\overline{h^{(-1)}_{22_{14}}}\right)\\
\nonumber
&-\left(\left|h^{(-1)}_{22_{11}}\right|^2+\left|h^{(-1)}_{22_{12}}\right|^2\right)\left(e^{j\theta}\overline{h^{(-1)}_{22_{21}}}h^{(-1)}_{22_{23}}+e^{-j\theta}h^{(-1)}_{22_{22}}\overline{h^{(-1)}_{22_{24}}}\right).
\end{align}}

Note that the entries of $H^{-1}_{22}$ are rational polynomial functions in the variables {\small$h^R_{22_{ij}}$} and {\small$h^I_{22_{ij}}$}, for $i,j=1,2,3,4$. If the expression in (\ref{eqn-final5}) is a non-constant rational polynomial function in {\small$h^R_{22_{ij}}$} and {\small$h^I_{22_{ij}}$} then, clearly (\ref{eqn-final5}) is non-zero almost surely, for any $\theta$. This is because, under a common denominator, the numerator of (\ref{eqn-final5}) would be a non-constant polynomial function in {\small$h^R_{22_{ij}}$} and {\small$h^I_{22_{ij}}$} which are independent and continuously distributed random variables for all $i,j$. To show that the expression in (\ref{eqn-final5}) is a non-constant rational polynomial function in {\small$h^R_{22_{ij}}$} and {\small$h^I_{22_{ij}}$} for some $(i,j)$ and for any $\theta$, it is sufficient to show that (\ref{eqn-final5}) evaluates to different values for different choices of $H_{22}$. Choose two values for $H_{22}$ to be 

{\footnotesize
\begin{align*}
&H_{22}=\begin{bmatrix}   
   0 &  0 &  1 & 0\\
   0 &  0 &  0 & 1\\
   1.5 & -1 & -0.5 &  -0.5\\
   -1 &  1  & 0  & 0
\end{bmatrix}, \begin{bmatrix}
0 &  0 &  1 & 0\\
   0 &  0 &  0 & 1\\
   1 & -0.5 & -0.5 &  -0.5\\
   -0.5 &  0.5  & 0  & 0
\end{bmatrix}
\end{align*}
}so that for the first matrix,

{\small \begin{align*}
&h^{(-1)}_{22_{11}}=h^{(-1)}_{22_{12}}=h^{(-1)}_{22_{21}}=h^{(-1)}_{22_{22}}=1,\\
& h^{(-1)}_{22_{13}}=h^{(-1)}_{22_{14}}=h^{(-1)}_{22_{23}}=2, ~h^{(-1)}_{22_{24}}=3\\
& H^{-1}_{22}(3,:)=[1 ~0 ~0 ~0], ~ H^{-1}_{22}(4,:)=[0 ~1 ~0 ~0].
\end{align*}}and for the second matrix all the entries of $H^{-1}_{22}$ are the same as above except that $h^{(-1)}_{22_{24}}=4$. Thus, for any value of $\theta$, (\ref{eqn-final5}) evaluates to $-e^{-j\theta}$ and $-2e^{-j\theta}$ for the two chosen values of $H_{22}$. Hence, for any value of $\theta$, the expression in (\ref{eqn-final5}) is a non-constant rational polynomial function in the entries of $H_{22}$. 
\end{proof}

\begin{lemma}\label{lem5}
The random variable $p$ defined in (\ref{eqn-define_p}) is non-zero almost surely.
\end{lemma}
\begin{proof}
We have 
{
\begin{align*} \label{eqn-simplify_p}
& p=\left(|e_1|^2+|e_2|^2\right)\left[e^{i\theta}\overline{v_{11_{12}}}v_{11_{14}}+e^{-i\theta}v_{11_{11}}\overline{v_{11_{13}}}\right].
\end{align*}}From Lemma \ref{lem4}, since $e_1$ and $e_2$ are non-zero almost surely, we only need to need to prove that {\small$e^{i\theta}\overline{v_{11_{12}}}v_{11_{14}}+e^{-i\theta}v_{11_{11}}\overline{v_{11_{13}}}$} is non-zero almost surely. Since {\small$V_{11}=\frac{H^{-1}_{12}}{\text{tr}\left(H^{-1}_{12}H^{-H}_{12}\right)}$}, we only need to show that {\small$e^{i\theta}\overline{h^{(-1)}_{12_{12}}}h^{(-1)}_{12_{14}}+e^{-i\theta}h^{(-1)}_{12_{11}}\overline{h^{(-1)}_{12_{13}}}$} is non-zero because {\small$\text{tr}\left(H^{-1}_{12}H^{-H}_{12}\right)$} is non-zero almost surely. Using similar arguments as in Lemma \ref{lem4}, it can be shown that {\small$e^{i\theta}\overline{h^{(-1)}_{12_{12}}}h^{(-1)}_{12_{14}}+e^{-i\theta}h^{(-1)}_{12_{11}}\overline{h^{(-1)}_{12_{13}}}$} is a non-constant rational polynomial function in the entries of $H_{12}$, for any $\theta$. Hence, {\small$e^{i\theta}\overline{h^{(-1)}_{12_{12}}}h^{(-1)}_{12_{14}}+e^{-i\theta}h^{(-1)}_{12_{11}}\overline{h^{(-1)}_{12_{13}}}$} is non-zero almost surely.
\end{proof}

Let us now complete the proof for the statement that the first row, first column entry of the matrix {\small$\left[\frac{H''^H_3 H''_1}{||H''_3(1,:)||^2}-\frac{H''^H_4 H''_2}{||H''_4(1,:)||^2}\right]$} is non-zero almost surely. The coefficients of $h^{ {\scriptscriptstyle R}^2}_{11_{31}}$ and $h^{ {\scriptscriptstyle I}^2}_{11_{31}}$ in the expression {\small$\frac{1}{||H''_3(1,:)||^2}\left[H''^H_3 H''_1-||H''_3(1,:)||^2\frac{H''^H_4 H''_2}{||H''_4(1,:)||^2}\right]_{11}$} can be derived to be equal to

{\footnotesize\begin{align*}
&p-\left(|e_1|^2+|e_2|^2\right)\left(\left|v_{12_{13}}\right|^2+\left|v_{12_{14}}\right|^2\right)\frac{H''^H_4 H''_2}{||H''_4(1,:)||^2} \text{ and } \\
&-p-\left(|e_1|^2+|e_2|^2\right)\left(\left|v_{12_{13}}\right|^2+\left|v_{12_{14}}\right|^2\right)\frac{H''^H_4 H''_2}{||H''_4(1,:)||^2}
\end{align*}}respectively. Clearly, since $p$ is non-zero almost surely, both of the above coefficients cannot be equal to zero simultaneously. Thus, {\small$\left[H''^H_3 H''_1-||H''_3(1,:)||^2\frac{H''^H_4 H''_2}{||H''_4(1,:)||^2}\right]_{11}$} is a quadratic polynomial in the continuously distributed random variables $h^{ {\scriptscriptstyle R}^2}_{11_{31}}$ and $h^{ {\scriptscriptstyle I}^2}_{11_{31}}$ and hence, non-zero almost surely.

\end{IEEEproof}

\end{document}